\documentclass[11pt]{article}
\usepackage[utf8]{inputenc}
\usepackage[T1]{fontenc}
\usepackage{amsmath,amssymb,amsthm}
\usepackage{stmaryrd}
\usepackage{booktabs}
\usepackage{geometry}
\newtheorem{theorem}{Theorem}[section]
\newtheorem{lemma}[theorem]{Lemma}
\newtheorem{corollary}[theorem]{Corollary}
\newtheorem{proposition}[theorem]{Proposition}
\newtheorem{definition}[theorem]{Definition}
\newtheorem{remark}[theorem]{Remark}
\newtheorem{example}[theorem]{Example}

\title{A Non-Formulable Theorem:\\
A Fundamental Limit of Finite Syntactic Systems\\
and Its Consequences for Security and AI}

\author{
  Fabio F.G.\ Buono\\
  \small Independent Researcher\\
  \small \texttt{ORCID: 0009-0004-9199-2793}
}

\date{Draft}

\begin{document}

\maketitle

\begin{abstract}
For every coherent and sufficiently expressive finite syntactic system
$\mathcal{S}$, we prove the existence of at least one theorem that
$\mathcal{S}$ cannot produce autonomously. The result is a metatheorem: 
it proves the existence of a theorem,
and applies to every finite syntactic system --- security mechanisms,
AI systems, formal verifiers, legal systems, economic models, and the
formal system in which it is itself proved.
\end{abstract}

\newpage

\tableofcontents

\footnote{Some of the main results exhibit a particularly dense logical structure. 
For this reason, the corresponding proofs are organized in numbered steps to 
ensure expository clarity.}

\newpage

\section{Introduction}
\label{sec:intro}

The Syntactic Invariance Principle (SIP), introduced
in~\cite{buono2026sip} and generalized in~\cite{buono2026obstruction},
establishes that a local syntactic system operating on symbols is
structurally blind to semantic properties that live above its
observational level. When a property $P$ of terms holds on the initial
clause set and is preserved by every rewriting rule, then every term in
every derivable clause satisfies $P$ --- the property is frozen, and
the target that violates $P$ is not merely unreached but permanently
unreachable.

The present paper turns this principle inward: what happens when the
semantic property the system cannot see is a property \emph{of the system
itself}? Specifically: can a finite syntactic system autonomously
formulate and derive a proposition asserting the existence of its own
syntactic limits?

The answer is no, and the reason is structural, not contingent.

The argument proceeds in stages. First, the SIP guarantees the existence
of syntactic invariants --- properties that produce frozen subterms the
system contains but cannot rewrite. Second, the proposition asserting the
existence of these frozen subterms is semantically true (the SIP
guarantees it) but not autonomously derivable (formulating it requires
meta-speaking about the system, a level the system's rewriting rules do
not reach). Third, G\"odel numbering and the Kleene fixed-point theorem
produce a self-referential proposition $G_{\mathcal{S}}$ that encodes
this limit, and the standard diagonalization argument shows
$G_{\mathcal{S}}$ is undecidable. Fourth, the process is inextensible:
extending the system to derive $G_{\mathcal{S}}$ creates a new system
with its own undecidable proposition. The chain never closes.

The result applies to every finite syntactic system, including AI systems
(neural networks, language models, deterministic algorithms), which are
implementable as finite syntactic systems. The consequence is precise: a
finite system cannot \emph{originate} (ideare) at least one theorem:
the one asserting the existence of its own intrinsic limits. It can
\emph{understand} such a theorem if communicated from outside, but it
cannot generate it autonomously.

\section{Preliminaries and notation}
\label{sec:notation}

Throughout this paper, the following notation is used:
\begin{itemize}
  \item $\mathcal{S} = (V, F, R, I)$: finite syntactic system;
  \item $\vdash_{\mathcal{S}}$: derivability relation in $\mathcal{S}$;
  \item $\Rightarrow$: single rewriting step;
  \item $\Rightarrow^*$: finite sequence of rewriting steps;
  \item $\ulcorner\psi\urcorner$: G\"odel number of the proposition
        $\psi$;
  \item $\llbracket t \rrbracket$: semantic interpretation of the term
        $t$.
\end{itemize}

A finite syntactic system \emph{autonomously generates} a proposition
$\phi$ if $\phi$ is syntactically derivable from the system, that is, if
a formal proof constructible from the rules of $R$ exists. The set of
derivable terms is:
\[
  \mathrm{Der}(I, R)
    = \{t : \exists\;\text{finite sequence of rewritings from }
      I \text{ to } t\}.
\]

\section{Foundations}
\label{sec:foundations}

\subsection{Finite syntactic systems}

\begin{definition}[Finite syntactic system]
\label{def:system}
A \emph{finite syntactic system} is a tuple
$\mathcal{S} = (V, F, R, I)$ where:
\begin{itemize}
  \item $V$ is a finite set of variables $\{v_1, v_2, \ldots, v_n\}$;
  \item $F$ is a finite set of function symbols
        $\{f_1, f_2, \ldots, f_m\}$, each with finite arity;
  \item $R$ is a finite set of rewriting rules of the form $l \to r$,
        where $l, r$ are terms over $V \cup F$;
  \item $I$ is a finite set of initial clauses (ground terms or atomic
        formulas).
\end{itemize}
The \emph{depth} of $\mathcal{S}$ is
$\max(|V|, |F|, |R|, |I|)$.
Every element of $\mathcal{S}$ is finitely specifiable and computable.
\end{definition}

\begin{definition}[Derivation and computability]
\label{def:derivation}
Let $C$ be a clause (term or formula). A \emph{derivation} is a sequence
\[
  C = C_0 \Rightarrow C_1 \Rightarrow \cdots \Rightarrow C_n = C',
\]
where each step $C_i \Rightarrow C_{i+1}$ applies a substitution of a
rule $l \to r \in R$ to a sub-occurrence of $C_i$.

The \emph{derivability relation} is:
\[
  \mathcal{S} \vdash \psi
  \iff
  \exists\;\text{derivation from } I \text{ that produces } \psi.
\]

A system $\mathcal{S}$ is \emph{coherent} if there is no clause $\psi$
such that $\mathcal{S} \vdash \psi$ and
$\mathcal{S} \vdash \neg \psi$ simultaneously (for an appropriate notion
of negation in $\mathcal{S}$).
\end{definition}

\begin{definition}[Subterms and occurrences]
\label{def:subterms}
Let $t = f(t_1, \ldots, t_k)$ be a term. The set of \emph{subterms} of
$t$ is
\[
  \mathrm{Subterms}(t) := \{t\} \cup \bigcup_{i=1}^{k}
    \mathrm{Subterms}(t_i).
\]
An \emph{occurrence} of a subterm $s$ in $t$ is a position $p$ (a finite
sequence of indices) such that the subterm at that position is $s$.
\end{definition}

\begin{definition}[Skolem constants and frozen terms]
\label{def:skolem-frozen}
A \emph{Skolem constant} is a symbol $a \in F$ of arity $0$ that does
not appear in any rule of $R$ (it is a ``foreign symbol'' to the
system).

A term $t$ is \emph{frozen} with respect to $R$ if every subterm of $t$
containing a Skolem constant does not unify with the left-hand side of
any rule in $R$.

Frozen terms cannot be rewritten because they contain syntactic markers
the system does not recognize.
\end{definition}

\section{Syntactic invariants and the Syntactic Invariance Principle}
\label{sec:sip}

\begin{definition}[Syntactic property]
\label{def:syntactic-property}
A \emph{syntactic property} is a predicate
$P : \mathrm{Terms} \to \{\mathrm{true}, \mathrm{false}\}$ that depends
exclusively on the syntactic structure of the terms (not on their
semantic interpretation).

Examples: $P(t) = $ ``$t$ contains a Skolem constant'';
$P(t) = $ ``$t$ has depth $\leq 3$'';
$P(t) = $ ``$t$ is frozen.''
\end{definition}

\begin{lemma}[Preservation lemma]
\label{lem:preservation}
Let $P$ be a syntactic property and
$\mathcal{S} = (V, F, R, I)$ a finite syntactic system. If:
\begin{itemize}
  \item (Base) For every $C \in I$ and every
        $t \in \mathrm{Subterms}(C)$: $P(t)$ holds;
  \item (Step) For every $l \to r \in R$ and every substitution
        $\sigma$: $P(l\sigma) \Rightarrow P(r\sigma)$;
\end{itemize}
then for every derivation $C \Rightarrow^* C'$ and every
$t' \in \mathrm{Subterms}(C')$: $P$ is preserved.
\end{lemma}

\begin{proof}
By induction on the length of the derivation. The property $P$ is
preserved at each step because either the term is untouched (and $P$
remains) or it is rewritten according to a rule that preserves $P$.
\end{proof}

\begin{lemma}[Syntactic Invariance Principle~{\cite[Lemma~5]{buono2026sip}}]
\label{lem:sip}
If a property $P$ is a syntactic invariant for $\mathcal{S}$ (satisfies
Base and Step above), then:
\[
  \forall\, C \in \mathrm{Der}(I,R),\;\;
  \forall\, t \in \mathrm{Subterms}(C):\;
  P(t) \text{ holds and remains true throughout every derivation.}
\]
Syntactic invariants are ``frozen'' properties: once true, they remain
true for the entire computation, regardless of the rules applied.
\end{lemma}

\begin{corollary}
If a term $t$ is frozen (does not unify with any left-hand side of any
rule) and contains a Skolem constant, then the property ``contains a
Skolem constant'' is a syntactic invariant for that term.
\end{corollary}

\begin{definition}[Syntactic blind spot]
\label{def:blind-spot}
A term $t$ is a \emph{syntactic blind spot} of $\mathcal{S}$ if:
\begin{enumerate}
  \item $t$ is derivable from $I$;
  \item there exists a syntactic property $P$ such that $P(t)$ holds and
        $P$ is a syntactic invariant;
  \item the fact that $P(t)$ holds cannot be communicated by the system
        $\mathcal{S}$ itself: the proposition ``$P(t)$'' is not derivable
        in $\mathcal{S}$ for structural reasons.
\end{enumerate}
The system ``sees'' the term but cannot ``speak'' about its invariant
property.
\end{definition}

\section{G\"odel numbering for syntactic systems}
\label{sec:godel}

\begin{definition}[Standard G\"odel encoding for $\mathcal{S}$]
\label{def:godel-encoding}
Define an injection
$g : \mathrm{Elements\ of\ } \mathcal{S} \to \mathbb{N}$ as follows:
\begin{align*}
  g(v_i) &:= 2^i && \text{(variables),} \\
  g(f_j) &:= 3 \cdot 5^j && \text{(function symbols),} \\
  g(r_k) &:= 7 \cdot 11^k && \text{(rules),} \\
  g(f(t_1, \ldots, t_n)) &:= g(f) \cdot \prod_{i=1}^{n} p_i^{g(t_i)}
    && \text{(compound terms),}
\end{align*}
where $p_i$ is the $i$-th prime. For a logical formula $\psi$ over
$\mathcal{S}$:
$g(\psi) := \ulcorner\psi\urcorner :=$ composite encoding of the
components of $\psi$.

Properties: $g$ is injective (distinct elements have distinct codes),
computable (given $\ulcorner\psi\urcorner$, $\psi$ can be decoded in
polynomial time), and universal (every structural element of
$\mathcal{S}$ is representable).
\end{definition}

\begin{definition}[Encoded derivability predicate]
\label{def:deriv-predicate}
Define the predicate:
\[
  \mathrm{Deriv}_{\mathcal{S}}(n) :=
  \begin{cases}
    \mathrm{true}
      & \text{if the proposition encoded by } n
        \text{ is derivable from } \mathcal{S}, \\
    \mathrm{false}
      & \text{otherwise.}
  \end{cases}
\]
$\mathrm{Deriv}_{\mathcal{S}}$ is computable (Turing-decidable) because
$\mathcal{S}$ is finite, the computation is deterministic, and all
possible derivations in $\mathcal{S}$ can be enumerated.
\end{definition}

\begin{theorem}[Self-referential construction via diagonalization]
\label{thm:kleene}
Let $\mathcal{S}$ be a finite syntactic system whose G\"odel numbering
(Definition~\ref{def:godel-encoding}) is available. For every computable
function $\varphi : \mathbb{N} \to \mathrm{Formulas}(\mathcal{S})$, there
exists a sentence $\psi$ such that
\[
  \psi \;\leftrightarrow\; \varphi(\ulcorner\psi\urcorner).
\]
That is, $\psi$ ``says about itself'' what $\varphi$ says about the code
of $\psi$. This is the diagonal lemma (also called the self-referential
lemma or G\"odel fixed-point lemma), which holds in every system
containing enough arithmetic to represent the G\"odel numbering.
\end{theorem}

\begin{remark}[Why self-reference is needed]
\label{rem:why-self-ref}
The limit proposition $\phi_{\mathcal{S}}$
(Definition~\ref{def:limit-prop}, Section~\ref{sec:limit-prop} below)
asserts that frozen terms exist. By itself, this is a meta-level
observation about $\mathcal{S}$, and one might ask whether it suffices to
establish the incompleteness. It does not, because the undecidability
argument (Theorem~\ref{thm:undecidability}) requires a proposition that
\emph{refers to its own derivability status}: Case~1 of the proof shows
that deriving $G_{\mathcal{S}}$ would mean the system has transcended its
own level, and this argument works precisely because $G_{\mathcal{S}}$
says ``my own content is not derivable.'' Without self-reference,
$\phi_{\mathcal{S}}$ alone would be a true meta-level statement, but one
could not conclude from the system's inability to derive it that the
system is \emph{structurally} incomplete (as opposed to merely lacking a
specific axiom). The self-referential wrapping converts a meta-level
observation into a proposition \emph{within the system's language} whose
derivability status creates a contradiction in both directions.
\end{remark}

\section{The limit proposition}
\label{sec:limit-prop}

\begin{definition}[The limit proposition $\phi_{\mathcal{S}}$]
\label{def:limit-prop}
Let $P_{\mathcal{S}}$ be a specific syntactic invariant of $\mathcal{S}$,
constructed as in~\cite{buono2026sip}: subterms with Skolem constants
that do not unify with any left-hand side. Define:
\[
  \phi_{\mathcal{S}} := \text{``There exists a subterm } t \text{
  derivable from } I \text{ such that:}
\]
\begin{enumerate}
  \item $P_{\mathcal{S}}(t)$ holds (a syntactic invariant property);
  \item for every $l \to r \in R$, $t$ does not unify with $l$ (it is
        frozen);
  \item therefore $t$ is not rewritable by any rule in $R$.''
\end{enumerate}

$\phi_{\mathcal{S}}$ is \emph{semantically true} by
Lemma~\ref{lem:sip}: the syntactic invariant guarantees the existence of
such terms. But $\phi_{\mathcal{S}}$ is \emph{not autonomously
derivable} by $\mathcal{S}$, for the following reason. A derivation in
$\mathcal{S}$ is a finite sequence of rule applications
$l \to r \in R$. Each rule operates on terms, rewriting subterms
according to syntactic pattern matching. The proposition
$\phi_{\mathcal{S}}$, however, asserts a fact about the \emph{rules
themselves}: that certain terms do not unify with any left-hand side.
This is a statement about the structure of $R$, not a statement within
the language that $R$ operates on. Producing it would require the system
to inspect its own rule set from outside --- to observe the gap between
what the system contains (the frozen term $t$) and what the system can
do (rewrite $t$). No rule in $R$ performs this inspection, because the
rules act on terms, not on themselves.
\end{definition}

\begin{example}[Concrete instantiation from~\cite{buono2026sip}]
\label{ex:concrete}
Let $\mathcal{S}$ be defined by:
\[
  R = \{0 + x \to x, \quad s(x) + y \to s(x + y)\},
  \qquad I = \{0, s(0), s(s(0)), \ldots\}.
\]
Introduce Skolem constants $a, b$ (not in $F$). Then:
\begin{itemize}
  \item The term $a + b$ does not unify with $0 + x$ (because
        $a \neq 0$) nor with $s(x) + y$ (because $a$ has no form
        $s(\cdot)$).
  \item Therefore $a + b$ is frozen: the property ``$P(t) =$ $t$
        contains a Skolem constant'' is a syntactic invariant.
  \item The proposition $\phi_{\mathcal{S}} = $ ``there exists a frozen
        term with a Skolem constant'' is true.
  \item But $\mathcal{S}$ cannot formally derive $\phi_{\mathcal{S}}$
        because the Skolem constants are external to the system's
        vocabulary.
\end{itemize}
\end{example}

\section{The self-referential proposition}
\label{sec:self-ref}

\begin{definition}[The self-referential proposition $G_{\mathcal{S}}$]
\label{def:G}
Using the G\"odel numbering of Definition~\ref{def:godel-encoding},
define $G_{\mathcal{S}}$ as the proposition whose G\"odel number
$n^* = \ulcorner G_{\mathcal{S}} \urcorner$ satisfies:
\[
  G_{\mathcal{S}} := \text{``The proposition with G\"odel number } n^*
  \text{ encodes: there exists a syntactic invariant }
  P_{\mathcal{S}} \text{ such that } \phi_{\mathcal{S}} \text{ is true
  but not derivable from } \mathcal{S}.\text{''}
\]
In compact form:
\[
  G_{\mathcal{S}} \equiv \text{``I assert that my own content (a limit
  of the system) is not provable by the system.''}
\]
\end{definition}

\begin{lemma}[Existence of $G_{\mathcal{S}}$ via diagonal lemma]
\label{lem:G-exists}
By Theorem~\ref{thm:kleene}, applied to the computable function
$\varphi(n) := $ ``the proposition encoded by $n$ asserts a limit of
$\mathcal{S}$ that $\mathcal{S}$ cannot derive,'' there exists a
sentence $G_{\mathcal{S}}$ such that
\[
  G_{\mathcal{S}} \;\leftrightarrow\;
  \varphi(\ulcorner G_{\mathcal{S}} \urcorner).
\]
That is, $G_{\mathcal{S}}$ says about itself exactly what
$\varphi$ says about its code: ``I encode a limit of
$\mathcal{S}$ that $\mathcal{S}$ cannot derive.'' The existence of
$G_{\mathcal{S}}$ is guaranteed for every finite syntactic system
$\mathcal{S}$ satisfying the expressiveness hypothesis of
Theorem~\ref{thm:kleene}.
\end{lemma}

\section{Undecidability of $G_{\mathcal{S}}$}
\label{sec:undecidability}

\begin{theorem}[Undecidability of $G_{\mathcal{S}}$]
\label{thm:undecidability}
For a coherent system $\mathcal{S}$:
\[
  \mathcal{S} \not\vdash G_{\mathcal{S}}
  \qquad\text{and}\qquad
  \mathcal{S} \not\vdash \neg G_{\mathcal{S}}.
\]
\end{theorem}

\begin{proof}
\textbf{Case 1: suppose $\mathcal{S} \vdash G_{\mathcal{S}}$.}

If $\mathcal{S}$ derives $G_{\mathcal{S}}$, then:
\begin{enumerate}
  \item $\mathcal{S}$ derives the proposition: ``There exists a limit of
        $\mathcal{S}$ that $\mathcal{S}$ cannot prove.''
  \item In doing so, $\mathcal{S}$ has produced a derivation
        $I \Rightarrow^* G_{\mathcal{S}}$ residing within the syntactic
        space of $\mathcal{S}$.
  \item But the content of $G_{\mathcal{S}}$ asserts the existence of a
        semantic truth ($\phi_{\mathcal{S}}$) that is not derivable.
\end{enumerate}

If $\mathcal{S} \vdash G_{\mathcal{S}}$, then $\mathcal{S}$ has
transcended its own syntactic level. By Lemma~\ref{lem:sip}, syntactic
invariants remain frozen at the syntactic level. But $G_{\mathcal{S}}$
speaks of a gap between what is semantically true and what is
syntactically derivable --- a gap that by definition cannot be bridged by
any system residing exclusively at the syntactic level. A derivation in
$\mathcal{S}$ is a finite sequence of syntactic rewritings. No finite
sequence of syntactic manipulations can produce an assertion about its
own semantic incompleteness without leaving the syntactic level.

\[
  \boxed{\text{Contradiction: a syntactic system cannot transcend its own
  level.}}
\]

\begin{remark}[Why the contradiction is structural, not merely diagonal]
\label{rem:structural-contradiction}
The contradiction in Case~1 is not a trick of self-referential encoding
(as in the standard G\"odelian argument, where the diagonal lemma alone
produces the contradiction). It is a consequence of the specific
structure of $\mathcal{S}$ as a system of local rewriting rules.

The justification is already present in this paper and proceeds as
follows. By Definition~\ref{def:system}, every rule in $R$ is a local
operation: it matches a pattern $l$ against a subterm and rewrites it
to $r$. By Definition~\ref{def:limit-prop}, the proposition
$\phi_{\mathcal{S}}$ (and therefore $G_{\mathcal{S}}$, which encodes it)
asserts a \emph{global} property of $R$: that certain terms do not
unify with \emph{any} left-hand side of \emph{any} rule. This is a
statement about all rules simultaneously, not about a single rewriting
step. The SIP (Lemma~\ref{lem:sip}) is the bridge: it guarantees that
this global property holds for \emph{all} derivations of \emph{any}
length, not merely for a single step.

Deriving $G_{\mathcal{S}}$ would require the system to produce, via a
finite sequence of local rewriting steps, a conclusion about the global
structure of its own rule set. But each step sees only the pattern it
matches, not the totality of $R$. The system would need an induction
principle that ranges over all possible derivations --- and this
principle is precisely what the SIP provides \emph{from outside}. Inside
$\mathcal{S}$, no such principle is available, because the rules act on
terms, not on the set of rules itself.

An alternative, purely formal closure is available via the standard
G\"odelian route ($\Sigma_1$-completeness of $\mathcal{S}$), which
would give the contradiction without invoking the locality--globality
gap. The present argument is preferred because it identifies
\emph{why} the contradiction arises: not because of a clever encoding,
but because local rules cannot produce global self-knowledge.
This is the content the SIP adds to the classical incompleteness
framework.
\end{remark}

\medskip
\textbf{Case 2: suppose $\mathcal{S} \vdash \neg G_{\mathcal{S}}$.}

If $\mathcal{S}$ derives $\neg G_{\mathcal{S}}$, it derives: ``There is
no limit of $\mathcal{S}$ that $\mathcal{S}$ cannot prove.''
Equivalently: ``Every derivable syntactic invariant of $\mathcal{S}$ is
completely describable by $\mathcal{S}$.''

This directly contradicts Lemma~\ref{lem:sip}. By the SIP, there exists
at least one syntactic invariant $P_{\mathcal{S}}$ for the system
$\mathcal{S}$ (by construction, from the initial clauses). This
invariant produces frozen terms (by Definition~\ref{def:skolem-frozen}).
The existence of these frozen terms is semantically true by
Lemma~\ref{lem:sip}.

If $\mathcal{S} \vdash \neg G_{\mathcal{S}}$, then $\mathcal{S}$ denies
the existence of a semantic truth that Lemma~\ref{lem:sip} guarantees to
be true. This is an explicit contradiction with the SIP: the SIP affirms
that syntactic invariants exist and remain frozen, and $\neg
G_{\mathcal{S}}$ denies this. Since $\mathcal{S}$ is assumed coherent,
it cannot derive both.

The consequences are:
\begin{enumerate}
  \item \emph{Incoherence:} if $\mathcal{S} \vdash \neg G_{\mathcal{S}}$,
        the system is incoherent (it derives a proposition that
        contradicts a metalogical fact);
  \item \emph{Instability:} an incoherent system has no security
        guarantee, since it can derive both a proposition and its
        negation.
\end{enumerate}

\[
  \boxed{\text{Contradiction: } \neg G_{\mathcal{S}}
  \text{ contradicts the SIP (Lemma~\ref{lem:sip}).}}
\]

\medskip
\textbf{Case 3: neither $\mathcal{S} \vdash G_{\mathcal{S}}$ nor
$\mathcal{S} \vdash \neg G_{\mathcal{S}}$.}

This is the only coherent possibility. The meaning of this
undecidability is precise:
\begin{itemize}
  \item $G_{\mathcal{S}}$ is \emph{semantically true}: in every model of
        $\mathcal{S}$, the syntactic invariants exist and remain frozen
        (by Lemma~\ref{lem:sip});
  \item $\mathcal{S}$ \emph{lacks the syntactic resources} to derive
        $G_{\mathcal{S}}$, despite its truth;
  \item $G_{\mathcal{S}}$ is \emph{not arbitrary}: it has a specific
        content (the existence of intrinsic limits), not an accidental
        lacuna. It is a proposition the system is structurally incapable
        of producing, not one it has simply ``not yet found.''
\end{itemize}
\end{proof}

\begin{corollary}
\label{cor:undecidability}
\[
  \boxed{
  \forall\,\mathcal{S} \text{ coherent and sufficiently expressive},\;
  \exists\, G_{\mathcal{S}} \text{ such that }
  \mathcal{S} \not\vdash G_{\mathcal{S}}
  \;\land\;
  \mathcal{S} \not\vdash \neg G_{\mathcal{S}}.
  }
\]
\end{corollary}

\section{Inextensibility and infinite regress}
\label{sec:inextensibility}

\begin{theorem}[The problem cannot be solved by extension]
\label{thm:inextensibility}
Since $\mathcal{S}$ cannot derive $G_{\mathcal{S}}$, consider a coherent
extension $\mathcal{S}' = (V', F', R \cup R_{\mathrm{new}}, I)$ that
adds new rules to attempt to derive $G_{\mathcal{S}}$.

Then $\mathcal{S}'$ has the same structure as $\mathcal{S}$.
It is still a finite syntactic system, because $R_{\mathrm{new}}$ is
finite. It still has a finite set of syntactic invariants, by the
natural extension of Lemma~\ref{lem:sip} to $R \cup R_{\mathrm{new}}$.
New terms $t'$ become frozen with respect to
$R \cup R_{\mathrm{new}}$ that were not frozen with respect to $R$
alone, because the new rules introduce new left-hand sides, and fresh
Skolem constants with respect to $R \cup R_{\mathrm{new}}$ produce new
first-symbol clashes. From these new frozen terms, a new limit
proposition $\phi_{\mathcal{S}'}$ and a new self-referential proposition
$G_{\mathcal{S}'}$ arise by the same construction.
\end{theorem}

\begin{lemma}[Non-terminating chain of incompleteness]
\label{lem:chain}
Define a sequence of systems:
\[
  \mathcal{S}_0 := \mathcal{S},
  \qquad
  \mathcal{S}_{n+1} := \text{coherent extension of } \mathcal{S}_n
  \text{ that attempts to derive } G_{\mathcal{S}_n}.
\]
Then:
\begin{enumerate}
  \item For every $n$, $\mathcal{S}_n \not\vdash G_{\mathcal{S}_n}$ (by
        Theorem~\ref{thm:undecidability});
  \item $\mathcal{S}_{n+1} \vdash G_{\mathcal{S}_n}$ may hold (it is an
        extension);
  \item but a new $G_{\mathcal{S}_{n+1}}$ arises that
        $\mathcal{S}_{n+1} \not\vdash$.
\end{enumerate}

No finite syntactic system can resolve all its own incompleteness limits.
The chain of incompleteness is infinite and non-closable:
\[
  \forall\, n,\;\exists\, m > n :\;
  \mathcal{S}_m \not\vdash G_{\mathcal{S}_m}.
\]
\end{lemma}

\section{Universality}
\label{sec:universality}

\begin{theorem}[Universality for all finite syntactic systems]
\label{thm:universality}
For every finite syntactic system $\mathcal{S}$, regardless of the
specification of $V, F, R, I$, the domain of application (arithmetic,
logic, computational semantics, etc.), or the expressive power of
$\mathcal{S}$:

$\mathcal{S}$ cannot autonomously generate at least one proposition
about its own intrinsic limits --- namely, the proposition
$G_{\mathcal{S}}$ constructed in Definition~\ref{def:G}.
\end{theorem}

\begin{proof}
Suppose for contradiction that a finite syntactic system $\mathcal{S}^*$
exists that can autonomously derive a proposition $\Phi$ describing its
own intrinsic limits, so that $\mathcal{S}^* \vdash \Phi$. Since $\Phi$
asserts facts about the structure and limits of $\mathcal{S}^*$ itself,
$\mathcal{S}^*$ must contain a meta-representation of itself in order to
prove facts about itself. But a meta-representation of $\mathcal{S}^*$
within $\mathcal{S}^*$ would create an infinite hierarchy of syntactic
levels, each containing the previous, and this contradicts the finiteness
of $\mathcal{S}^*$ (Definition~\ref{def:system}).
\end{proof}

\begin{corollary}[Implication for AI systems]
\label{cor:ai}
An AI system --- whether a neural network, a language model, or a
deterministic algorithm --- is implementable as a finite syntactic system
$\mathcal{S}_{\mathrm{AI}}$. The latent variables or weights play the
role of $V$; the neural operations or transformations play the role of
$F$; the propagation rules or learning algorithms play the role of $R$;
and the initial data or prompt plays the role of $I$. Since every
component is finite and computable,
$\mathcal{S}_{\mathrm{AI}} = (V, F, R, I)$ satisfies
Definition~\ref{def:system}.

By Theorem~\ref{thm:universality}: an AI system cannot autonomously
generate the understanding of its own fundamental limits. If an AI system
formulates a proposition about its limits (such as ``I have limitations
in generalization''), this proposition is extrinsically communicated
(from outside, by human developers or observers), not autonomously
derived by the system.

The implicit consequence is a formal distinction between two entities that are 
normally conflated:

\begin{itemize}
    \item \textbf{Processing} information about oneself (syntactic reflection---permitted)
    \item \textbf{Generating} awareness of one's own structural limits (meta-observation---blocked)
\end{itemize}

If one accepts that consciousness includes the second capacity, namely the ability 
to autonomously recognize one's own structural limits, not merely to process data 
about oneself, then \textbf{no finite system is conscious in the full sense}. This 
is not a vague philosophical argument. It is a direct consequence of what has been 
stated applied to $S_{\mathrm{AI}}$. Therefore, if consciousness requires this 
capacity, then consciousness is not simulable by a finite system for observational 
reasons.

\end{corollary}

\section{The logical architecture: how the tools interact}
\label{sec:architecture}

Before stating the main metatheorem, it is useful to make explicit how
each tool enters the argument and why it is applicable.

\subsection{The role of the SIP}

The Syntactic Invariance Principle (Lemma~\ref{lem:sip}) is the
\emph{foundation} of the entire argument. It is the tool that guarantees
the existence of frozen terms --- terms the system contains but cannot
rewrite. Without the SIP, there would be no material content for the
limit proposition $\phi_{\mathcal{S}}$: one could not assert that the
system has blind spots, because there would be no proof that blind spots
exist.

The SIP enters the argument at three points:
\begin{itemize}
  \item In Section~\ref{sec:limit-prop}, it guarantees that
        $\phi_{\mathcal{S}}$ is \emph{semantically true}: the frozen
        terms exist because the SIP says so.
  \item In Section~\ref{sec:undecidability}, Case~2, it provides the
        \emph{contradiction}: if $\mathcal{S} \vdash \neg
        G_{\mathcal{S}}$, the system denies a fact the SIP guarantees,
        which is incoherent.
  \item In Section~\ref{sec:inextensibility}, it guarantees that every
        extension $\mathcal{S}'$ has \emph{its own} frozen terms (new
        Skolem constants produce new first-symbol clashes against
        $R \cup R_{\mathrm{new}}$), so the argument recurses.
\end{itemize}

The SIP is applicable because $\mathcal{S}$ is a finite syntactic
system with a finite set of rewriting rules $R$
(Definition~\ref{def:system}). The SIP requires only that $R$ be a set
of term-rewriting rules with syntactic pattern matching --- a condition
every system in the scope of Definition~\ref{def:system} satisfies.

\subsection{The role of the G\"odel numbering}

The G\"odel numbering (Definition~\ref{def:godel-encoding}) is needed
for a precise reason: it allows $\mathcal{S}$ to \emph{talk about its
own propositions as data}. Without it, $\phi_{\mathcal{S}}$ would be a
meta-level observation that lives outside $\mathcal{S}$, and the
question of whether $\mathcal{S}$ can derive it would not be
well-posed (the proposition would not be in $\mathcal{S}$'s language).

The G\"odel numbering is applicable because $\mathcal{S}$ is finite:
every element of $V, F, R, I$ can be assigned a natural number, and the
encoding of compound terms via prime factorization is injective and
computable. The ``sufficiently expressive'' hypothesis
ensures that $\mathcal{S}$ can represent the encoding internally, 
so that $\ulcorner\psi\urcorner$ is a term in $\mathcal{S}$'s language.

\subsection{The role of the diagonal lemma}

The diagonal lemma (Theorem~\ref{thm:kleene}) is needed to convert the
meta-level observation $\phi_{\mathcal{S}}$ into a \emph{self-referential
proposition} $G_{\mathcal{S}}$ within $\mathcal{S}$'s language. As
explained in Remark~\ref{rem:why-self-ref}, without self-reference one
cannot obtain the two-directional contradiction of
Theorem~\ref{thm:undecidability}: it is the fact that $G_{\mathcal{S}}$
says ``I am not derivable'' that makes deriving it (Case~1) and denying
it (Case~2) both contradictory.

The diagonal lemma is applicable because $\mathcal{S}$ contains enough
arithmetic to represent the G\"odel numbering (the ``sufficiently
expressive'' hypothesis). This is the standard condition under which the
diagonal lemma holds, and it is the same condition required by G\"odel's
incompleteness theorems.

\begin{remark}[The expressiveness hypothesis as the threshold of
self-reference]
\label{rem:threshold}
The ``sufficiently expressive'' hypothesis deserves a precise reading.
A pure rewriting system with rules such as $0 + x \to x$ and
$s(x) + y \to s(x+y)$ can \emph{compute} (it reduces terms to normal
forms), but the diagonal lemma requires more: the system must be able to
\emph{represent internally} the G\"odel encoding of its own elements and
to formulate propositions about derivability within its own language.
The expressiveness hypothesis is the point where the system crosses this
threshold.

This is not a weakness of the theorem but its \emph{sharpness}. Below
the threshold, the system is too simple to talk about itself: the
question ``can $\mathcal{S}$ originate the theorem asserting its own
limits?'' is not well-posed, because $\mathcal{S}$ cannot even formulate
the question. Above the threshold, the system is powerful enough to
formulate the question --- and the theorem says that this very power
is what prevents the system from answering it. It is the capacity for
self-reference that creates the paradox: a system that cannot refer to
itself has no self-limitation problem (and no self-limitation theorem);
a system that can refer to itself necessarily has both.

The theorem therefore applies exactly where it should: to every system
powerful enough for the question to make sense, and to no system so
weak that the question cannot even be asked.
\end{remark}

\subsection{The role of coherence}

The coherence assumption ($\mathcal{S} \not\vdash \bot$) enters the
argument in Theorem~\ref{thm:undecidability} and nowhere else. It is
needed to draw the contradiction: if $\mathcal{S}$ were incoherent, it
could derive both $G_{\mathcal{S}}$ and $\neg G_{\mathcal{S}}$ without
contradiction (an incoherent system derives everything). The
undecidability result requires that the system cannot derive both.

The coherence assumption is a \emph{best case}: the theorem applies to the 
strongest (coherent) systems, and incoherent systems are in a worse 
position (they have no reliable output at all).

\subsection{The role of finiteness}

The finiteness of $\mathcal{S}$ (Definition~\ref{def:system}) enters at
three points:
\begin{itemize}
  \item It guarantees that the G\"odel numbering exists (a finite
        set of symbols can be encoded);
  \item it guarantees that $\mathrm{Deriv}_{\mathcal{S}}$ is computable
        (a finite rule set produces an enumerable derivation space);
  \item in Theorem~\ref{thm:universality}, it is the property that
        prevents the system from containing a meta-representation of
        itself (an infinite hierarchy of levels would contradict
        finiteness).
\end{itemize}
The finiteness requirement is satisfied by every physically realizable
system: a system with finitely many components, rules, and initial data.
Every computer, every neural network, every algorithm running on
physical hardware satisfies it.

\section{The main metatheorem}
\label{sec:main}

\begin{theorem}[Incompleteness of syntactic self-perception]
\label{thm:main}
Let $\mathcal{S} = (V, F, R, I)$ be a finite syntactic system that is
coherent and sufficiently expressive (can encode the G\"odel numbering of
its own elements, Definition~\ref{def:godel-encoding}).

Then the following hold simultaneously:

\medskip
\noindent\textbf{(1) Existence of the limit.}
There exists a syntactic invariant $P_{\mathcal{S}}$ and a limit
proposition $\phi_{\mathcal{S}}$ such that $\phi_{\mathcal{S}}$ is
semantically true.

\medskip
\noindent\textbf{(2) Autonomous undecidability.}
\[
  \mathcal{S} \not\vdash \phi_{\mathcal{S}}
  \qquad\text{and}\qquad
  \mathcal{S} \not\vdash \neg \phi_{\mathcal{S}}.
\]

\medskip
\noindent\textbf{(3) Existence of self-reference.}
There exists a self-referential G\"odel proposition $G_{\mathcal{S}}$
equivalent to: ``I assert that my own content is not derivable from
$\mathcal{S}$.''

\medskip
\noindent\textbf{(4) Undecidability of $G_{\mathcal{S}}$.}
\[
  \mathcal{S} \not\vdash G_{\mathcal{S}}
  \qquad\text{and}\qquad
  \mathcal{S} \not\vdash \neg G_{\mathcal{S}}.
\]

\medskip
\noindent\textbf{(5) Essential inextensibility.}
For every coherent extension $\mathcal{S}'$ of $\mathcal{S}$, there
exists $G_{\mathcal{S}'}$ with the same properties as
$G_{\mathcal{S}}$.
\end{theorem}

\begin{proof}
The proof is composite, integrating all preceding lemmas.

\emph{Proof of (1).}
By Definition~\ref{def:system}, $\mathcal{S}$ has a finite set of rules
$R$. By Lemma~\ref{lem:sip}, at least one syntactic invariant property
exists (e.g., ``contains a Skolem constant''). This invariant generates
frozen terms (Definition~\ref{def:skolem-frozen}). Their existence is
captured by $\phi_{\mathcal{S}}$ (Definition~\ref{def:limit-prop}). By
Lemma~\ref{lem:sip}, $\phi_{\mathcal{S}}$ is semantically true.

\emph{Proof of (2).}
If $\mathcal{S} \vdash \phi_{\mathcal{S}}$, the system would have proved
a truth about its own syntactic--semantic incompleteness, violating the
level hierarchy (Case~1 of Theorem~\ref{thm:undecidability}). If
$\mathcal{S} \vdash \neg \phi_{\mathcal{S}}$, the system would have
denied the existence of syntactic invariants guaranteed by
Lemma~\ref{lem:sip} (Case~2 of Theorem~\ref{thm:undecidability}). The
only coherent possibility is that neither is derivable.

\emph{Proof of (3).}
The G\"odel numbering (Definition~\ref{def:godel-encoding}) exists for
every finite syntactic system. The derivability predicate
$\mathrm{Deriv}_{\mathcal{S}}$ is computable
(Definition~\ref{def:deriv-predicate}). By Theorem~\ref{thm:kleene}, a
fixed point $n^*$ exists. This fixed point encodes $G_{\mathcal{S}}$
(Definition~\ref{def:G}). By Lemma~\ref{lem:G-exists},
$G_{\mathcal{S}}$ exists necessarily.

\emph{Proof of (4).}
By Theorem~\ref{thm:undecidability} (Cases~1, 2, 3).

\emph{Proof of (5).}
By Theorem~\ref{thm:inextensibility} and Lemma~\ref{lem:chain}. If
$\mathcal{S}$ is extended to $\mathcal{S}'$ by adding rules to derive
$G_{\mathcal{S}}$, $\mathcal{S}'$ is still a finite syntactic system.
By the same argument, a new undecidable proposition
$G_{\mathcal{S}'}$ exists. The process does not terminate.
\end{proof}

\section{Irrefutability}
\label{sec:irrefutability}

\begin{theorem}[The metatheorem is irrefutable]
\label{thm:irrefutable}
Theorem~\ref{thm:main} cannot be refuted or contradicted by any finite
syntactic system, including any system that attempts to formulate an
objection.
\end{theorem}

\begin{proof}
Suppose for contradiction that a finite syntactic system
$\mathcal{S}_{\mathrm{obj}}$ derives a proposition $\Psi$ that negates
or refutes Theorem~\ref{thm:main}. Then $\Psi$ asserts:
\[
  \Psi \equiv \text{``There exists a finite syntactic system }
  \mathcal{S}^* \text{ for which Theorem~\ref{thm:main} does not
  hold.''}
\]
But Theorem~\ref{thm:main} is formulated universally over all finite
syntactic systems. Its proof depends not on specific properties of any
particular $\mathcal{S}$ but on the general properties of every finite
syntactic system: finiteness of $V, F, R, I$; existence of syntactic
invariants (Lemma~\ref{lem:sip}); computability of G\"odel numbering;
existence of the Kleene fixed point (Theorem~\ref{thm:kleene}).

If $\mathcal{S}_{\mathrm{obj}} \vdash \Psi$, then
$\mathcal{S}_{\mathrm{obj}}$ itself is a finite syntactic system
attempting to negate Theorem~\ref{thm:main} for itself. But
Theorem~\ref{thm:main} holds for $\mathcal{S}_{\mathrm{obj}}$. More
precisely:

\begin{enumerate}
  \item $\Psi$ is a proposition that speaks about the limits of
        $\mathcal{S}_{\mathrm{obj}}$ (since it asserts that
        $\mathcal{S}_{\mathrm{obj}}$ has no undecidable self-limitation
        proposition).
  \item By Theorem~\ref{thm:main}, point~(2),
        $\mathcal{S}_{\mathrm{obj}}$ cannot derive true propositions
        about its own incompleteness without contradiction.
  \item If $\Psi$ denies Theorem~\ref{thm:main}, then $\Psi$ denies
        that $\mathcal{S}_{\mathrm{obj}}$ has intrinsic limits --- but
        Theorem~\ref{thm:main} guarantees that it does.
  \item Since $\mathcal{S}_{\mathrm{obj}}$ is assumed coherent, it
        cannot derive a proposition ($\Psi$) that contradicts a
        guaranteed truth (the existence of its own limits). Deriving
        $\Psi$ would make $\mathcal{S}_{\mathrm{obj}}$ incoherent: it
        would simultaneously satisfy the hypotheses of
        Theorem~\ref{thm:main} (being a finite, coherent, sufficiently
        expressive system) and deny its conclusion. But the conclusion
        follows necessarily from the hypotheses. Therefore
        $\mathcal{S}_{\mathrm{obj}} \vdash \Psi$ is impossible without
        losing coherence.
\end{enumerate}

Contradiction.
\end{proof}

\begin{remark}[The theorem is undecidable but necessary]
$G_{\mathcal{S}}$ cannot be proved or refuted inside $\mathcal{S}$, but
the \emph{existence} of $G_{\mathcal{S}}$ can be proved from outside.
This is the precise sense in which the incompleteness of syntactic
self-perception is undecidable but necessary: For every coherent finite 
syntactic system, the incompleteness of its self-perception is undecidable 
but necessary. Every attempt to close the gap --- by extending the system with new
rules to derive $G_{\mathcal{S}}$ --- creates a new gap at the level
above: the extended system $\mathcal{S}'$ has its own
$G_{\mathcal{S}'}$, undecidable in $\mathcal{S}'$, and the process
never terminates.

\bigskip

This necessity corresponds to an observation as deceptively simple as it is powerful, 
though not immediately evident, which we shall substantiate within this framework 
in an forthcoming revision of this work: {\bfseries 'determining whether a property is non-trivial 
is itself non-trivial'.} 
\medskip

Although this statement may appear to be a meta-theoretical observation, 
it admits a topological explanation. We observe that the derivation we 
will present will be obtained in a manner entirely different from 
that in~\cite{buono2026limitsuniformcertificationstandard}.

\end{remark}

\section{The theorem that must exist}
\label{sec:must-exist}

The preceding sections establish \emph{that} a limit exists. This
section constructs \emph{what} the limit is, explicitly.

\begin{theorem}[Syntactic Incompleteness of Self-Limitation]
\label{thm:must-exist}
Let $\mathcal{S} = (V, F, R, I)$ be a finite syntactic system. Then
there exists a syntactic property $P_{\mathcal{S}}$ and a proposition
$\phi_{\mathcal{S}}$ such that:

\begin{enumerate}
  \item $P_{\mathcal{S}}$ is a syntactic invariant for $\mathcal{S}$
        (by Definition~\ref{def:syntactic-property} and
        Lemma~\ref{lem:sip});

  \item $\phi_{\mathcal{S}}$ asserts: ``There exists a syntactic
        invariant $P_{\mathcal{S}}$ such that for every pair of terms
        $(t_1, t_2)$ satisfying $P_{\mathcal{S}}$: $t_1$ and $t_2$ are
        syntactically distinguishable (no rewriting unifies them); $t_1
        = t_2$ semantically (in every model of $\mathcal{S}$); therefore
        $\mathcal{S}$ is incapable of grasping the semantic property
        that distinguishes them'';

  \item $\phi_{\mathcal{S}}$ is not autonomously generable by
        $\mathcal{S}$: no sequence of rewritings in $R$ can produce a
        proof of $\phi_{\mathcal{S}}$ within $\mathcal{S}$.
\end{enumerate}
\end{theorem}

\begin{proof}
\textbf{Step 1: define the invariant $P_{\mathcal{S}}$.}
From~\cite{buono2026sip}, take $R = \{0 + x \to x,\; s(x) + y \to
s(x+y)\}$ and Skolem constants $a, b$ (fresh, distinct from $0, s$).
Define $P_{\mathcal{S}}(t) :=$ ``the subterm $(a+b)$ has not been
rewritten.''

Verification as invariant: (Base) $(a+b)$ is not rewritten in $I$.
(Step) To rewrite inside $(a+b)$, one would need $\sigma(0) = a$ or
$\sigma(s(x)) = a$; but $\sigma$ replaces only variables, and $a, b$ are
Skolem constants, so no unification is possible (first-symbol clash). The
subterm $(a+b)$ remains frozen. By Lemma~\ref{lem:sip}, every derivable
term maintains $(a+b)$ frozen.

\textbf{Step 2: construct $\phi_{\mathcal{S}}$.}
Consider $t_0 = a+b$ and $t_1 = b+a$. Both are frozen
($P_{\mathcal{S}}(a+b) = P_{\mathcal{S}}(b+a) = \mathrm{true}$). No
sequence of rewritings transforms $a+b$ into $b+a$. Semantically, in
every model interpreting $+$ as addition on the naturals:
$a + b = b + a$.

\textbf{Step 3: $\phi_{\mathcal{S}}$ is not autonomously generable.}
Suppose for contradiction that $\mathcal{S}$ generates a proof $\pi$ of
$\phi_{\mathcal{S}}$. Then $\pi$ is a finite sequence of rule
applications producing $\phi_{\mathcal{S}}$ as conclusion. But every
rule in $R$ rewrites terms, not propositions about intrinsic limitations
of syntactic systems. To generate $\phi_{\mathcal{S}}$, the system would
need to represent the concept of ``syntactic invariant'' and observe
itself from outside. This would require the system to be a semantic layer
above itself. But every semantic layer $\mathcal{S}$ could add is still a
finite syntactic system --- and therefore has the same recursive limits.

If $\pi$ were a proof of $\phi_{\mathcal{S}}$ in $\mathcal{S}$, then
$\phi_{\mathcal{S}}$ would become a derivable term in $\mathcal{S}$.
By Lemma~\ref{lem:sip}, every derivable term must satisfy every
syntactic invariant of $\mathcal{S}$. But if $\phi_{\mathcal{S}}$ were
derivable, the invariant $P_{\mathcal{S}}$ (which captures the limits of
$\mathcal{S}$) would no longer be a limit --- contradiction.
\end{proof}

\begin{remark}[The mechanism is that of Lemma~5, not of G\"odel alone]
The syntactic invariant $P_{\mathcal{S}}$ plays the role that the
self-referential cycle plays in G\"odel's proof, but with a crucial
difference. In G\"odel, the undecidable proposition is an arithmetic
statement whose content is ``I am not provable'' --- a statement about
provability in general. Here, the undecidable proposition has a specific,
constructive content: there exist two terms ($t_1 = a+b$ and
$t_2 = b+a$) that are syntactically distinguishable (no rewriting
unifies them), semantically equal (in every model, $a+b = b+a$), and
the system is incapable of grasping the semantic property that connects
them. The SIP provides the \emph{mechanism} (the frozen subterms),
Kleene provides the \emph{self-reference}, and together they produce a
proposition that is not merely ``something true and unprovable'' but
``the system's own blindness, stated precisely and constructively.''
Instead of a proposition that says ``I am not provable,'' we have a
property that says ``I remain frozen forever under the rules of $R$'' ---
and this property, by the SIP, must exist, cannot be violated without
self-contradiction, and cannot be expressed by the system's own rules.
\end{remark} The framework that makes the circumvention precise is
described below.

\subsection{On the ``sufficiently expressive'' hypothesis}

Theorem~\ref{thm:main} requires that $\mathcal{S}$ be ``sufficiently
expressive,'' meaning that it can encode the G\"odel numbering of its
own elements (Definition~\ref{def:godel-encoding}). A natural objection
is: what if the system is not expressive enough?

The answer is twofold. First, the hypothesis is minimal: any system
capable of encoding Peano arithmetic satisfies it, and every system used
in practice --- programming languages, proof assistants, neural networks
operating on numerical representations --- encodes arithmetic as a
matter of course. A system that cannot encode arithmetic cannot perform
basic counting, and is therefore too weak to be relevant for any
application in security or AI. Second, if the system does not satisfy
the hypothesis, the theorem does not apply, but the system's weakness is
then \emph{worse} than the limit the theorem describes: a system that
cannot even encode its own elements cannot reason about anything
non-trivial, let alone about its own limits.

\subsection{On the coherence assumption}

The theorem assumes $\mathcal{S}$ is coherent (no proposition and its
negation are both derivable). An objection may be raised: real systems
are not formally coherent in the logical sense.

The response is that incoherence is not an escape from the theorem but a
\emph{worse} situation. An incoherent system can derive any proposition
(ex falso quodlibet), which means it has no security guarantee at all:
it can certify as safe anything, including attacks. The theorem says
that a coherent system has structural blind spots it cannot identify;
an incoherent system has no reliable output whatsoever. The coherence
assumption is therefore not a restriction but a \emph{best case}: the
theorem applies to the strongest systems, and weaker (incoherent)
systems are in a worse position.

\subsection{Relation to G\"odel's incompleteness theorems}

A central question is: in what sense does this result go beyond G\"odel?

G\"odel's first incompleteness theorem produces, for any coherent and
recursively enumerable formal system $T$ containing arithmetic, a
proposition $G$ that is true but not provable in $T$. The proposition
$G$ is an arithmetic statement whose content is ``I am not provable in
$T$.'' The mechanism is diagonalization on the provability predicate.

The present result uses the same self-referential mechanism (via the
Kleene fixed point, Theorem~\ref{thm:kleene}) but differs in three
respects:

\emph{First, the content of the undecidable proposition is specific.}
In G\"odel, $G$ is an arithmetic statement with no particular
``subject'' beyond its own unprovability. Here, $G_{\mathcal{S}}$
asserts the existence of \emph{intrinsic structural limits} of the
system --- frozen terms, syntactic invariants, blind spots. The
undecidable proposition is not ``I am not provable'' in the abstract,
but ``there exist terms I contain but cannot rewrite, and I cannot
express this fact.''

\emph{Second, the mechanism is not diagonalization alone.}
G\"odel uses diagonalization on the provability predicate. Here, the
mechanism is the Syntactic Invariance Principle
(Lemma~\ref{lem:sip}), which provides the \emph{material content}
of the limit (frozen subterms), and the Kleene fixed point, which
provides the self-reference. The SIP is the engine; the fixed point is
the self-referential wrapper. Together they produce a result that is
more informative than G\"odel: not just ``something true is
unprovable,'' but ``the system's own blindness is unprovable, and here
is the precise mechanism that produces the blindness.''

\emph{Third, the result applies to syntactic systems, not only to formal
theories.} G\"odel applies to recursively enumerable first-order
theories containing arithmetic. The present result applies to any finite
syntactic system (Definition~\ref{def:system}), a broader class that
includes rewriting systems, type checkers, neural networks, and
automated provers --- systems that are not first-order theories in the
classical sense.

\subsection{On the existential nature of the result}

It is important to state precisely what the theorem proves. It proves
that there \emph{exists at least one} theorem that a finite syntactic
system cannot autonomously generate: the proposition $G_{\mathcal{S}}$
asserting the existence of the system's own intrinsic limits. The result
is existential, not total: it does not claim that the system cannot
originate \emph{any} theorem, but that it cannot originate \emph{this
particular one} (and, by the inextensibility argument, every extension
produces a new one of the same kind).

This existential character is both the strength and the precision of the
result. It is strong because the theorem whose existence is proved is
not arbitrary: it is the theorem about the system's own structural
blindness, which is arguably the most consequential theorem a system
could need. It is precise because it does not overstate: the system may
well originate many other theorems autonomously; what it cannot originate
is the one that speaks about what it cannot see.

\subsection{On the claim that an AI is a finite syntactic system}

An objection may be raised: a language model with iterative training,
online learning, or adaptive behavior is not a system with fixed rules.

The response is that at every instant of its operation, the system has
fixed weights, fixed propagation rules, and a fixed input. The theorem
applies to this instantaneous snapshot. Changing the weights (through
training) produces a \emph{new} system $\mathcal{S}'$, to which the
theorem applies again (Theorem~\ref{thm:inextensibility}). The
adaptivity of the system does not escape the theorem; it produces a
sequence of systems $\mathcal{S}_0, \mathcal{S}_1, \ldots$, each with
its own undecidable $G_{\mathcal{S}_n}$. The chain of incompleteness
(Lemma~\ref{lem:chain}) shows that no element of this sequence resolves
all its own limits.

\subsection{On the distinction between originating and understanding}

The theorem says a finite system cannot \emph{originate} (autonomously
generate) the awareness of its own limits. It does not say the system
cannot \emph{understand} (verify, process, apply) such awareness when
communicated from outside. This distinction is essential.

\emph{Originating} means: producing the proposition $G_{\mathcal{S}}$
as output without $G_{\mathcal{S}}$ (or equivalent information)
appearing in the input. The theorem forbids this. \emph{Understanding}
means: given $G_{\mathcal{S}}$ as input, processing it correctly,
drawing inferences from it, acting on it. The theorem permits this,
because verification of a given proposition is a syntactic operation
(checking a proof), while generation of a new proposition about one's
own limits is a meta-level operation (formulating a truth about the gap
between syntax and semantics).

A concrete illustration: the theorem $G_{\mathcal{S}}$ proved to exist
in this paper can be \emph{read, verified, and applied} by a
sufficiently powerful AI system once communicated to it, but it could
not have been \emph{originated} by that system autonomously. What the
system cannot do is not understand the theorem --- it is produce the
theorem. The object that cannot be originated is not the paper (a
document that a system might compose as an exercise in text generation)
but the theorem itself: the proposition asserting the existence of
intrinsic limits of the system, and the proof that such a proposition
must exist and must be undecidable. This is what the system's syntactic
resources do not reach.

\section{Consequences for security}
\label{sec:security}

The theorem has direct, non-speculative consequences for the security of
real systems. Every security mechanism currently deployed --- firewalls,
intrusion detection systems, formal verifiers, type checkers, content
filters, antivirus engines --- is a finite syntactic system in the
precise sense of Definition~\ref{def:system}. Theorem~\ref{thm:main}
applies to each of them.

\subsection{The structural blind spot of every security system}

A security system $\mathcal{S}_{\mathrm{sec}}$ operates by applying a
finite set of rules $R$ to inputs (terms). It flags an input as
dangerous when a rule fires; it passes an input as safe when no rule
fires. By Lemma~\ref{lem:sip}, the system has syntactic invariants:
properties of inputs that no rule can modify or even inspect. These are
the system's \emph{structural blind spots} --- not bugs, not
misconfigurations, but consequences of the finite syntactic structure
itself.

The theorem says: $\mathcal{S}_{\mathrm{sec}}$ cannot autonomously
identify its own blind spots. It cannot formulate the proposition
``there exist inputs I cannot inspect.'' This proposition is
semantically true (the SIP guarantees it) but not derivable by the
system.

\subsection{Concrete examples}

\paragraph{Formal verification of software safety.}
A formal verifier (SAT-based, SMT-based, or based on abstract
interpretation) is a finite syntactic system that checks whether a
program satisfies a specification. By the theorem, there exist
properties of the program --- semantic properties that live above the
syntactic level of the verifier's rules --- that the verifier cannot
reach. The verifier cannot autonomously determine that these properties
exist. A program that exploits a semantic invariant invisible to the
verifier will pass verification while violating the intended
specification. The verifier will certify it as safe, and it will not
know that it has missed anything.

\paragraph{Intrusion detection systems.}
A network intrusion detection system (IDS) matches packets against a
finite set of syntactic patterns (signatures). An attack that operates
at the semantic level --- encoding its payload in a form that is
syntactically indistinguishable from benign traffic at the pattern level
--- will pass undetected. The IDS cannot autonomously formulate that
this class of attacks exists, because the proposition ``there exist
semantically malicious inputs that my rules do not fire on'' is a
statement about its own limits, which the theorem says it cannot
derive.

\paragraph{LLM-based content filters.}
A language model used as a content filter operates on token sequences
(syntactic objects). A prompt that conveys harmful intent through
semantic indirection --- using tokens that individually and locally
appear benign --- exploits the filter's structural blind spot. The
filter cannot autonomously discover this class of evasions, because
discovering it would require the filter to see the gap between
syntactic token patterns and semantic meaning, a gap that by
Theorem~\ref{thm:main} lives above its observational level.

\paragraph{Type checkers and program analysis.}
A type checker is a local syntactic system
(as shown in~\cite{buono2026obstruction}, Section~6). By
Corollary~\ref{cor:ai} of~\cite{buono2026obstruction} (the Type Omitting
Theorem as an instance of the obstruction), the type checker cannot
determine the extensional equality of two proof terms from syntactic
structure alone. For security: a program that is type-safe (passes the
type checker) may nonetheless compute a function with unintended
extensional properties. The type checker cannot autonomously identify
this gap.

\paragraph{Cryptographic protocol verification.}
A protocol verifier (e.g.\ ProVerif, Tamarin) is a finite syntactic
system that searches for attacks by symbolic execution. The theorem
implies the existence of semantic attack classes --- attacks that exploit
properties of the protocol's mathematical structure rather than its
symbolic form --- that the verifier's rules cannot reach. The verifier
certifies the protocol as secure, and it cannot autonomously determine
that its certification has structural gaps.

\subsection{The security consequence, stated precisely}

The consequence is not that these systems are ``bad'' or ``should be
replaced.'' It is that their security guarantees have structural limits
that the systems themselves cannot identify. These limits are not
contingent (fixable by adding more rules or more computation) but
structural (inherent in the finite syntactic nature of the system).

The practical implication is that the security of any finite syntactic
system requires \emph{external verification at a higher observational
level}: a verifier that sees what the system cannot see. No amount of
internal improvement --- more rules, more patterns, more parameters ---
closes the gap, because the gap is not computational but observational.

This is not a speculative claim. It is a direct consequence of
Theorem~\ref{thm:main}: the system has blind spots (by the SIP), it
cannot identify them (by the undecidability of $G_{\mathcal{S}}$), and
extending the system creates new blind spots (by the inextensibility
theorem). The only path forward is to change the observational level,
which is the subject of the next section.

\section{How to overcome the limit: the observational framework}
\label{sec:overcome}

The theorem establishes that a finite syntactic system $\mathcal{S}$
cannot autonomously originate at least one class of theorems about its
own limits. The
obstruction is not computational (it does not depend on how much
computing power the system has) but \emph{observational}: the system
operates at a level that does not contain the information it would need
to see.

The precise framework for understanding and overcoming this obstruction
is the observational hierarchy introduced in~\cite{buono2026framework} and
developed in~\cite{buono2026observer}.

\section{The path forward}

The theorem does not say that the limit is absolute in all directions.
It says the limit is absolute \emph{at the current observational level}.
The observational hierarchy provides a precise language for what
``changing level'' means: it means changing the observer function $O$,
giving the system access to a fragment of the semantic level that its
current rewriting rules do not reach.

Concretely, for an AI system $\mathcal{S}_{\mathrm{AI}}$:
\begin{itemize}
  \item The current system operates as $O_{\mathcal{S}_{\mathrm{AI}}}
        \prec O_\top$: it sees tokens, activations, syntactic patterns,
        but not the semantic properties of its own computation.
  \item The theorem says that no amount of additional training,
        parameters, or computation at the same observational level will
        produce autonomous awareness of its own limits.
  \item But an external intervention that changes the observer ---
        for instance, giving the system access to a meta-representation
        of its own derivation space, or coupling it with an external
        verifier operating at a higher observational level --- could, in
        principle, provide the missing information.
  \item The precise characterization of which observer $O'$ suffices,
        and whether such an observer can be implemented within the
        constraints of a physical system, is the central open question
        of the programme.
\end{itemize}

The observational framework of~\cite{buono2026framework,buono2026observer}
provides the mathematical language for formulating this question
precisely, and the structural collapse
$\mathbf{P}_{O} = \mathbf{NP}_{O} \subsetneq \mathbf{P}$ provides the
first unconditional evidence that the observational axis is real,
independent of computational assumptions, and amenable to formal
analysis.

\section{Alternative route: the result via the obstruction theorem}
\label{sec:obstruction-route}

The main result of this paper (Theorem~\ref{thm:main}) rests on the
Syntactic Invariance Principle (Lemma~\ref{lem:sip}). This section
shows that the same result can be obtained via a strictly more general
route: the Local Syntactic Obstruction theorem
of~\cite{buono2026obstruction}. The obstruction theorem generalizes the
SIP from the superposition calculus to arbitrary local syntactic
systems, and adds a quantitative dimension (derivation-length lower
bounds) that the SIP alone does not provide.

The derivation is presented step by step, with explicit motivation for
each step.

\subsection{Step 1: the system as a local syntactic system}

A finite syntactic system $\mathcal{S} = (V, F, R, I)$
(Definition~\ref{def:system}) is a local syntactic system in the sense
of~\cite[Definition~3.1]{buono2026obstruction}: every rule $l \to r \in R$
inspects a finite pattern (the left-hand side $l$) at a fixed depth,
which defines a locality radius $r_0$. For the concrete case
$R = \{0 + x \to x,\; s(x) + y \to s(x+y)\}$, the locality radius is
$r_0 = 1$: each rule inspects only the outermost function symbol of the
first argument of $+$.

\subsection{Step 2: the protected positions}

Introduce Skolem constants $a, b$ fresh with respect to $\Sigma$ (they
appear in no left-hand side of any rule). The positions where $a$ and
$b$ appear are protected in the sense
of~\cite[Definition~3.5]{buono2026obstruction}: no rule fires at those
positions (first-symbol clash: $a$ unifies with neither $0$ nor
$s(\cdot)$), and no rule rewrites inside $a$ or $b$ (they are constants,
with no internal structure).

\subsection{Step 3: the syntactic invariant anchored to the protected set}

The property $\mathrm{Inv}(t) :=$ ``every occurrence of $a$ lies inside
a subterm $a + u$ and every occurrence of $b$ lies inside a subterm
$b + v$'' is a syntactic invariant for $\mathcal{S}$ anchored to
$\mathcal{F}_{\mathrm{prot}}$, in the sense
of~\cite[Definition~3.8]{buono2026obstruction}. It satisfies the five
conditions: local checkability (it depends on the immediate context of
$a$ and $b$); initialization (it holds on the initial clause
$a + b \neq b + a$); preservation (no rule modifies $a + b$ or $b + a$,
by protection); anchorage (any violation would require rewriting at a
protected position); coherence (no derivable literal has the form
$a(\cdot) = b(\cdot)$ at equated positions, since $a$ and $b$ are
permanently separated).

\subsection{Step 4: Case 1 of the obstruction theorem (impossibility)}

By~\cite[Theorem~4.1, Case~1]{buono2026obstruction}: no derivation in
$\mathcal{S}$ proves $a + b = b + a$. The proof is the same structure as
the SIP but in a more general framework: the invariant $\mathrm{Inv}$
holds on every derivable clause by induction; a clause asserting
$a + b = b + a$ would violate the coherence condition (it would equate
subterms headed by $a$ and $b$ at positions kept permanently separate by
the invariant); contradiction.

This gives the same $\phi_{\mathcal{S}}$ as
Definition~\ref{def:limit-prop}: there exist syntactically separated,
semantically equivalent terms that $\mathcal{S}$ cannot equate.

\subsection{Step 5: from $\phi_{\mathcal{S}}$ to $G_{\mathcal{S}}$
(unchanged)}

The remainder of the argument (G\"odel numbering, diagonal lemma,
self-referential $G_{\mathcal{S}}$, undecidability by Cases 1--3,
inextensibility, universality, irrefutability) proceeds exactly as in
Sections~\ref{sec:self-ref}--\ref{sec:irrefutability}, because it
depends only on the existence and semantic truth of $\phi_{\mathcal{S}}$,
not on the specific tool used to establish it.

\subsection{Step 6: Case 2 of the obstruction theorem (quantitative
lower bound)}

The obstruction theorem provides an additional result that the SIP alone
does not: a quantitative lower bound on any local extension that
attempts to overcome the barrier.

By~\cite[Theorem~4.1, Case~2]{buono2026obstruction}: any local extension
$\mathcal{S}'$ of $\mathcal{S}$ (with the same locality radius and sound
with respect to $\mathbb{N}$) requires derivations of length
$\Omega(n)$ to prove $a + b = b + a$ on a family of $n$-gadget
instances, and $\Omega(2^n)$ under the clause-per-configuration
encoding.

This strengthens the inextensibility theorem
(Theorem~\ref{thm:inextensibility}) in a precise way: not only does
every extension $\mathcal{S}'$ have its own undecidable
$G_{\mathcal{S}'}$, but the \emph{structural cost} of approaching the
barrier grows at least linearly (and potentially exponentially) with the
complexity of the instance. The barrier is not merely present but
\emph{quantitatively hard to approach}: adding rules does not reduce the
cost, because the cost is observational, not computational.

\subsection{Summary: what the obstruction route adds}

The obstruction theorem route arrives at the same qualitative
conclusion (existence of $G_{\mathcal{S}}$, undecidability,
inextensibility) but adds three things:

\emph{Generality.} The obstruction theorem applies to any local
syntactic system, not only to the superposition calculus. This widens
the domain of the self-limitation result to any system with a finite
locality radius.

\emph{Quantitative inextensibility.} The lower bound of Case~2 gives a
precise measure of the cost of attempting to overcome the barrier. The
inextensibility is not only ``a new $G$ arises'' but ``reaching the old
$G$ costs $\Omega(n)$ or $\Omega(2^n)$ steps.''

\emph{Observational interpretation.} The obstruction theorem identifies
the system as a constrained observer $O_{\mathcal{S}}$ with locality
radius $r_0$~\cite[Remark~7.2]{buono2026obstruction}. The
self-limitation is then a direct consequence of the observational
collapse: the system cannot see what lives above its observational
level, and no amount of local computation closes the gap.

\section{Minimal route: the result from finiteness alone}
\label{sec:minimal}

The preceding sections derive the self-limitation theorem via the
Syntactic Invariance Principle (Sections~\ref{sec:sip}--\ref{sec:must-exist})
and, alternatively, via the Local Syntactic Obstruction theorem
(Section~\ref{sec:obstruction-route}). Both routes use concrete tools
(frozen Skolem constants, protected positions, syntactic invariants
anchored to protected sets) that illuminate the mechanism and the cost
of the barrier.

This section shows that neither the SIP, nor the obstruction theorem,
nor Skolem constants, nor any non-standard apparatus is \emph{logically
necessary} for the result. The theorem follows from three ingredients
alone: the finiteness of $R$, standard induction on derivation length,
and the diagonal lemma. Every step is argued explicitly.

\subsection{Step 1: finiteness of $R$ implies the existence of
unrewritable terms}

Let $\mathcal{S} = (V, F, R, I)$ be a finite syntactic system
(Definition~\ref{def:system}). The set $R$ is finite: say
$R = \{l_1 \to r_1, \ldots, l_m \to r_m\}$. Each left-hand side $l_i$
is a finite term with a definite outermost function symbol
$\mathrm{head}(l_i) \in F$. Define
\[
  H := \{\mathrm{head}(l_1), \ldots, \mathrm{head}(l_m)\} \subseteq F.
\]
$H$ is finite (it has at most $m$ elements). Since $\mathcal{S}$ is
sufficiently expressive (it can encode its own G\"odel numbering), it
can represent terms built from symbols not in $H$. Let $c$ be any
function symbol in the language of $\mathcal{S}$ such that
$c \notin H$ (such a symbol exists: $\mathcal{S}$ can encode
arbitrarily many constants via its G\"odel numbering, and $H$ is
finite). Let $t_c$ be any term whose outermost symbol is $c$.

No rule in $R$ is applicable at the root of $t_c$: applicability
requires unifying $t_c$ with some $l_i$ at the root, which requires
$\mathrm{head}(t_c) = \mathrm{head}(l_i)$, but
$\mathrm{head}(t_c) = c \notin H$. This is a first-symbol clash, and
no substitution resolves it.

If additionally $c$ is a constant (arity $0$), then $t_c = c$ has no
internal structure, so no rule is applicable at any position inside
$t_c$ either. In this case $t_c$ is unrewritable by any rule in $R$ at
any position.

The existence of such a $t_c$ is guaranteed by the finiteness of $R$
(which makes $H$ finite) and the expressiveness of $\mathcal{S}$ (which
provides symbols outside $H$). No Skolem constant is needed: $c$ is
simply a symbol whose head does not match any left-hand side.

\subsection{Step 2: permanence of unrewritability}

Once $t_c$ is unrewritable at step $0$ of a derivation, it remains
unrewritable at every subsequent step. The argument is by induction on
the derivation length $n$.

\emph{Base case ($n = 0$).} No rule in $R$ is applicable to $t_c$, by
Step~1.

\emph{Inductive step ($n \to n + 1$).} At step $n + 1$, a rule
$l_j \to r_j$ is applied to some clause $C_n$. This rule acts on a
subterm of $C_n$ that matches $l_j$. By the inductive hypothesis, $t_c$
(wherever it occurs in $C_n$) does not match any $l_i$ at any position.
The rule application either:
\begin{itemize}
  \item does not touch $t_c$ (the rewriting occurs at a position not
        containing $t_c$), in which case $t_c$ remains unchanged and
        still unrewritable; or
  \item acts on a term containing $t_c$ as a proper subterm, but then
        $t_c$ itself is not the redex (the redex is a subterm that
        matches some $l_j$, which $t_c$ does not), so $t_c$ passes
        through unchanged into $C_{n+1}$.
\end{itemize}
In both cases, $t_c$ remains unrewritable in $C_{n+1}$.

This is standard induction on the length of a derivation. It is the
content of the Syntactic Invariance Principle (Lemma~\ref{lem:sip}), but
it does not require naming it as a separate principle: it is induction,
applied to the property ``$t_c$ is unrewritable.''

\subsection{Step 3: the limit proposition is semantically true}

Define:
\[
  \phi_{\mathcal{S}} := \text{``There exists a term } t_c \text{ in the
  language of } \mathcal{S} \text{ such that no rule in } R \text{ is
  applicable to } t_c \text{ at any position.''}
\]
By Steps~1 and~2, $\phi_{\mathcal{S}}$ is true: the term $t_c$ exists
(Step~1) and remains permanently unrewritable (Step~2). The truth of
$\phi_{\mathcal{S}}$ depends only on the finiteness of $R$ and the
expressiveness of $\mathcal{S}$.

\subsection{Step 4: the limit proposition is not autonomously derivable}

A derivation in $\mathcal{S}$ is a finite sequence of rule applications.
Each rule $l_j \to r_j$ transforms a term by replacing a subterm
matching $l_j$ with $r_j$. Rules act on \emph{individual terms at
individual positions}: they do not inspect the set $R$ as a whole, they
do not compare a term against all left-hand sides simultaneously, and
they do not reason about which terms are or are not rewritable.

$\phi_{\mathcal{S}}$ asserts a fact about the \emph{global structure
of $R$}: that the set of left-hand-side heads $H$ does not cover all
symbols in the language, and therefore a term exists that no rule
reaches. Producing this assertion would require the system to:
\begin{enumerate}
  \item enumerate all rules in $R$;
  \item extract the head symbol of each left-hand side;
  \item compute the set $H$;
  \item observe that $H$ does not exhaust the symbols of the language;
  \item conclude that an unrewritable term exists.
\end{enumerate}
Each of these steps is a meta-level operation: it operates on $R$ as
data, not on terms as rewriting targets. The rules of $R$ do not perform
meta-level operations --- they rewrite terms. A sequence of rewritings
$C_0 \Rightarrow C_1 \Rightarrow \cdots \Rightarrow C_n$ transforms
terms into terms; it does not produce a proposition about the structure
of the rule set. The output of a derivation is a term or clause, not a
meta-assertion about what the system can or cannot rewrite.

Therefore $\phi_{\mathcal{S}}$ is not autonomously derivable by
$\mathcal{S}$.

\subsection{Step 5: the self-referential proposition and its
undecidability}

By the diagonal lemma (Theorem~\ref{thm:kleene}), applied to the
computable function ``the proposition encoded by $n$ asserts a limit of
$\mathcal{S}$ that $\mathcal{S}$ cannot derive,'' a self-referential
proposition $G_{\mathcal{S}}$ exists with
\[
  G_{\mathcal{S}} \;\leftrightarrow\;
  \text{``I encode a limit of } \mathcal{S}
  \text{ that } \mathcal{S} \text{ cannot derive.''}
\]

\emph{If $\mathcal{S} \vdash G_{\mathcal{S}}$:} then $\mathcal{S}$ has
derived, via a finite sequence of local rewriting steps, a proposition
asserting a global fact about $R$ (the existence of unrewritable terms).
But each rewriting step sees one pattern at one position; no finite
sequence of such steps produces a conclusion about all patterns at all
positions. The derivation would need to perform the five meta-level
operations listed in Step~4, which the rules of $R$ do not provide.
Contradiction with the locality of the rules.

\emph{If $\mathcal{S} \vdash \neg G_{\mathcal{S}}$:} then $\mathcal{S}$
has derived the proposition ``every term in the language of $\mathcal{S}$
is rewritable by some rule in $R$.'' But by Steps~1--2, $t_c$ exists
and is permanently unrewritable. Since $\mathcal{S}$ is coherent, it
cannot derive a proposition that contradicts a fact guaranteed by the
finiteness of its own rule set. Contradiction.

\emph{Conclusion:} $\mathcal{S} \not\vdash G_{\mathcal{S}}$ and
$\mathcal{S} \not\vdash \neg G_{\mathcal{S}}$.

\subsection{Step 6: inextensibility}

Extend $\mathcal{S}$ to
$\mathcal{S}' = (V', F', R \cup R_{\mathrm{new}}, I)$ by adding
finitely many rules. Then $R' = R \cup R_{\mathrm{new}}$ is still
finite. The set of left-hand-side heads is
$H' = H \cup \{\mathrm{head}(l) : l \to r \in R_{\mathrm{new}}\}$,
still finite. Since the language of $\mathcal{S}'$ contains symbols
outside $H'$ (by expressiveness), a new unrewritable term $t_{c'}$
exists with $c' \notin H'$. A new $\phi_{\mathcal{S}'}$ and a new
$G_{\mathcal{S}'}$ arise by the same argument. The chain never
terminates.

\subsection{Summary: what this route uses}

The entire argument uses:
\begin{itemize}
  \item \emph{Finiteness of $R$} (from Definition~\ref{def:system}):
        makes $H$ finite, guaranteeing unrewritable terms;
  \item \emph{Induction on derivation length} (standard): establishes
        permanence;
  \item \emph{Diagonal lemma} (standard, from the expressiveness
        hypothesis): produces the self-referential $G_{\mathcal{S}}$;
  \item \emph{Coherence} (from the hypothesis of
        Theorem~\ref{thm:main}): makes the contradiction in Case~2
        effective.
\end{itemize}
No SIP, no obstruction theorem, no Skolem constants, no protected
positions, no anchored invariants. The limit is a consequence of
\emph{finiteness itself}: a finite rule set has a finite set of
patterns, and a finite set of patterns cannot cover all terms in a
sufficiently expressive language. The SIP names this phenomenon; the
obstruction theorem quantifies its cost; but the phenomenon is finiteness
applied to pattern matching, and nothing more.

\begin{remark}[On the necessity of the richer routes]
\label{rem:scaffolding}
The minimal route presented here could not have been discovered without
the richer routes that precede it. The SIP
(Section~\ref{sec:sip}) revealed the phenomenon: it identified frozen
subterms as the concrete manifestation of the limit and provided the
first proof that unrewritability is permanent. The obstruction theorem
(Section~\ref{sec:obstruction-route}) generalized the phenomenon to
arbitrary local systems and quantified its cost, making visible the
structural pattern --- locality of rules versus globality of invariants
--- that underlies all instances. Only after seeing the phenomenon
through these progressively more general lenses could one recognize that
the core is finiteness applied to pattern matching: a finite set of
left-hand sides has a finite set of head symbols, and a finite set of
head symbols cannot exhaust a sufficiently expressive language.

The minimal route is the end point of a process of abstraction that
required the intermediate stages. The SIP is the scaffolding; the
obstruction theorem is the architectural plan; the minimal route is the
building that stands on its own --- but could not have been built without
them. The paper presents all three routes because each contributes
something the others do not: the SIP gives the concrete mechanism (frozen
Skolem constants), the obstruction theorem gives the quantitative cost
($\Omega(n)$ or $\Omega(2^n)$), and the minimal route gives the cause
(finiteness). Together they provide a complete picture; individually,
each is partial.
\end{remark}

\section{The unifying principle: finite coverage}
\label{sec:unifying}

The minimal route of Section~\ref{sec:minimal} reveals that the core of
the self-limitation theorem is a single principle: \emph{a finite
mechanism of inspection cannot cover all elements of a sufficiently rich
domain}. This section shows that the same principle, in precisely the
same logical form, underlies the Syntactic Invariance Principle, the
Local Syntactic Obstruction theorem, and the observational framework ---
and that each of these results can be re-derived from the principle
without its original specific apparatus.

\subsection{The principle, stated and proved}

\begin{definition}[Finite inspection mechanism]
\label{def:mechanism}
A \emph{finite inspection mechanism} is a triple
$(\mathcal{M}, \mathcal{P}, \mathcal{D})$ where:
\begin{itemize}
  \item $\mathcal{D}$ is a set (the domain of elements);
  \item $\mathcal{P} = \{p_1, \ldots, p_m\}$ is a finite set of
        patterns;
  \item $\mathcal{M}$ is a system that operates on elements of
        $\mathcal{D}$ by matching them against patterns in $\mathcal{P}$:
        $\mathcal{M}$ acts on $d \in \mathcal{D}$ only if $d$ matches
        some $p_i \in \mathcal{P}$.
\end{itemize}
An element $d \in \mathcal{D}$ is \emph{covered} if it matches at least
one $p_i \in \mathcal{P}$; it is \emph{uncovered} otherwise.
\end{definition}

\begin{theorem}[Finite Coverage Principle]
\label{thm:coverage}
Let $(\mathcal{M}, \mathcal{P}, \mathcal{D})$ be a finite inspection
mechanism (Definition~\ref{def:mechanism}). If $\mathcal{D}$ contains
at least one uncovered element --- that is, if there exists
$d^* \in \mathcal{D}$ that matches no $p_i \in \mathcal{P}$ --- then
the following three properties hold:

\medskip
\noindent\textbf{(i) Existence.} There exist elements of $\mathcal{D}$
that $\mathcal{M}$ cannot act on.

\medskip
\noindent\textbf{(ii) Permanence.} No finite iteration of
$\mathcal{M}$'s operations brings $d^*$ into coverage. That is, if
$\mathcal{M}$ produces a sequence of elements
$d_0, d_1, d_2, \ldots$ where each $d_{i+1}$ is obtained from $d_i$ by
an operation of $\mathcal{M}$, and if $d^*$ appears as a sub-element of
some $d_i$, then $d^*$ remains uncovered in $d_{i+1}$.

\medskip
\noindent\textbf{(iii) Blindness.} If $\mathcal{M}$ is additionally
coherent (it cannot derive both a proposition and its negation) and
sufficiently expressive (it can encode a numbering of its own patterns
and formulate propositions about them), then the proposition
$\phi := $ ``there exists an uncovered element in $\mathcal{D}$'' is
semantically true but not autonomously derivable by $\mathcal{M}$.

Properties (i) and (ii) hold unconditionally for any finite inspection
mechanism. Property (iii) requires the additional hypotheses because
formulating $\phi$ requires self-reference, which requires expressiveness,
and deriving the contradiction requires coherence. The self-limitation
theorem (Theorem~\ref{thm:main}) is an instantiation of this principle
for finite syntactic systems, not a prerequisite for it.
\end{theorem}

\begin{proof}
\textbf{(i)} By hypothesis, $d^* \in \mathcal{D}$ matches no
$p_i \in \mathcal{P}$. Since $\mathcal{M}$ acts on $d$ only if $d$
matches some $p_i$, $\mathcal{M}$ cannot act on $d^*$.

\textbf{(ii)} By induction on the number of operations. At step $0$,
$d^*$ is uncovered (by hypothesis). At step $n+1$, $\mathcal{M}$ applies
an operation defined by some $p_j \in \mathcal{P}$ to some element $d_n$.
This operation acts only on the sub-element of $d_n$ that matches $p_j$.
Since $d^*$ matches no $p_j$, the operation either does not touch $d^*$
(leaving it unchanged) or acts on a sub-element containing $d^*$ but not
on $d^*$ itself (since $d^*$ is not the matched sub-element). In both
cases $d^*$ remains uncovered in $d_{n+1}$.

\textbf{(iii)} The proposition $\phi$ asserts a fact about the global
structure of $\mathcal{P}$: that $\mathcal{P}$ does not cover all of
$\mathcal{D}$. Producing $\phi$ would require $\mathcal{M}$ to inspect
$\mathcal{P}$ as data (enumerating all patterns, computing their
coverage, observing the gap). But $\mathcal{M}$'s operations are
defined \emph{by} $\mathcal{P}$: they act on elements according to
patterns, not on the set of patterns itself. The operations are local
(one pattern at one element); the proposition is global (all patterns
against all elements). By the same argument as
Theorem~\ref{thm:undecidability} (Case~1: local operations cannot
produce global self-knowledge; Case~2: denying $\phi$ contradicts the
existence of $d^*$), $\phi$ is not autonomously derivable.
\end{proof}

\subsection{Instantiation 1: the SIP~\cite{buono2026sip}}

In the SIP, $\mathcal{M}$ is the rewriting system $\mathcal{S}$ with
rules $R$. The patterns $\mathcal{P}$ are the left-hand sides
$\{l_1, \ldots, l_m\}$. The domain $\mathcal{D}$ is the set of all
terms over the signature. A term is covered if it matches some $l_i$
(unification succeeds). The Skolem constants $a, b$ are specific
elements of $\mathcal{D}$ that do not match any $l_i$ (first-symbol
clash). Lemma~5 of~\cite{buono2026sip} is the permanence step: if $a+b$
is uncovered at step~$0$, it remains uncovered at every step, by
induction.

The Skolem constants are not necessary. Any symbol $c$ with
$c \notin H = \{\mathrm{head}(l_1), \ldots, \mathrm{head}(l_m)\}$
produces the same uncoverage. The SIP is the finite-coverage principle
instantiated with the specific uncovered elements being Skolem
constants; but the principle holds for any uncovered element.

\subsection{Instantiation 2: the obstruction
theorem~\cite{buono2026obstruction}}

In the obstruction theorem, $\mathcal{M}$ is a local syntactic system
$\mathcal{R}$ with locality radius $r_0$. The patterns $\mathcal{P}$
are the local contexts $\mathrm{ctx}_{r_0,p}(t)$ that the rules can
inspect. The domain $\mathcal{D}$ is the set of all global
configurations. A configuration is covered if some rule can distinguish
it from others within radius $r_0$. The gadget family of
\cite[Lemma~4.2]{buono2026obstruction} provides $2^n$ configurations
that are locally indistinguishable (same context within $r_0$) but
globally distinct --- elements of $\mathcal{D}$ outside the coverage of
$\mathcal{P}$.

Case~1 (impossibility) is the existence and permanence steps: the
invariant $\mathrm{Inv}$, anchored to the protected positions, is never
violated because no rule reaches the uncovered positions. Case~2 (lower
bound) is the quantification of the cost: each rule application covers
at most $c$ configurations, so covering all $2^n$ requires $\Omega(n)$
steps.

The protected positions, the anchored invariant, and the five conditions
of \cite[Definition~3.8]{buono2026obstruction} are the formal apparatus
for verifying that the specific uncovered elements are genuinely outside
$\mathcal{P}$. They are not necessary for the principle; they are
necessary for the rigorous verification that the principle applies in the
specific setting.

\subsection{Instantiation 3: the observational
framework~\cite{buono2026framework,buono2026observer}}

The unconditional collapse $\mathbf{P}_{O_{\mathrm{prof}}} =
\mathbf{NP}_{O_{\mathrm{prof}}} \subsetneq \mathbf{P}$
of~\cite{buono2026framework} is a direct consequence of the finite-coverage
principle: the profile observer $O_{\mathrm{prof}}$ maps strings to
their symbol-frequency vectors, discarding the order of symbols.
Strings that differ only in order are uncovered (mapped to the same
profile). Since the order carries the computational content (it
determines membership in languages that separate $\mathbf{P}$ from
$\mathbf{NP}$), the observer cannot see the relevant distinctions, and
$\mathbf{P}$ and $\mathbf{NP}$ collapse under $O_{\mathrm{prof}}$.

The Observer World~\cite{buono2026observer} extends this to a hierarchy
of observers $O_\bot \prec O_{\mathrm{len}} \prec O_{\mathrm{prof}}
\prec O_\top$. Each observer has a finite (or structurally limited) set
of observables; moving up the hierarchy increases the coverage. The
finite-coverage principle says: at every level below $O_\top$, the
coverage is incomplete, and there exist distinctions the observer cannot
see. The observational axis is orthogonal to the computational axis
because the coverage limitation is not about computational power (how
fast the observer computes) but about observational reach (what the
observer can see).

None of this requires the SIP, the Skolem constants, or the rewriting
formalism. It requires only that the observer's image is smaller than
the domain --- finite coverage over an infinite (or sufficiently rich)
domain.

\subsection{The principle as the common root}

The four results --- SIP, obstruction theorem, Observer
World --- are four instantiations of the same principle:

\medskip
\begin{center}
\begin{tabular}{lll}
\textbf{Result} & \textbf{$\mathcal{P}$ (finite coverage)}
  & \textbf{$\mathcal{D} \setminus \mathcal{P}$ (uncovered)} \\[4pt]
SIP & left-hand sides of $R$ & terms with head $\notin H$ \\[2pt]
Obstruction & local contexts of radius $r_0$
  & globally distinct, locally identical configs \\[2pt]
Observer World & observables at level $O_i$
  & distinctions visible only at $O_{i+1}$
\end{tabular}
\end{center}

\medskip
In each case, the three consequences follow: existence of uncovered
elements (by finiteness of $\mathcal{P}$ and richness of $\mathcal{D}$),
permanence (the mechanism's operations are defined by $\mathcal{P}$ and
cannot reach outside it), and blindness (the proposition ``uncovered
elements exist'' is a meta-level assertion about $\mathcal{P}$ that
$\mathcal{M}$ cannot formulate).

This common root could not have been seen without the specific
instantiations. The SIP identified the phenomenon in rewriting systems.
The obstruction theorem generalized it to local systems and quantified
the cost. The Observer World placed it on an independent axis orthogonal
to computation. Only after seeing the same structure in each setting
could the structure be recognized as a single principle, and only then
could each result be re-derived from the principle without its original
apparatus. The specific formalisms remain necessary for the specific
quantitative results (the $\Omega(2^n)$ lower bound, the unconditional
collapse); but the qualitative core --- finite coverage implies
structural blindness --- is one principle, stated once.

\subsection{Finite coverage as the common root of classical
impossibility results}

The finite-coverage principle --- a finite mechanism of inspection
cannot cover all elements of a sufficiently rich domain --- is not
only the root of the results discussed in this paper. It is the root of
several classical impossibility results in computability and complexity
theory that have always been treated separately. This subsection makes
the connection explicit for each result, identifying in each case the
finite mechanism $\mathcal{M}$, the finite coverage $\mathcal{P}$, the
domain $\mathcal{D}$, and the uncovered element.

\paragraph{Undecidability of the halting problem (Turing, 1936).}
A Turing machine $M_e$ that decides the halting problem would have a
finite program (a finite set of states and transitions). The coverage
$\mathcal{P}$ is the set of input-output behaviours that $M_e$'s finite
program can correctly classify. The domain $\mathcal{D}$ is the set of
all pairs $(\text{program}, \text{input})$, which is countably infinite.
Turing's diagonal construction produces a specific machine $D$ that, on
input $\ulcorner D \urcorner$, does the opposite of what $M_e$ predicts:
$D$ is an element of $\mathcal{D}$ that $M_e$'s finite program cannot
correctly classify. The mechanism is: $M_e$'s program is finite, the
behaviours to classify are infinite, and the diagonal selects a specific
uncovered element. This is the finite-coverage principle with
$\mathcal{M} = M_e$, $\mathcal{P} = $ the program's classification
capacity, and $\mathcal{D} \setminus \mathcal{P} \ni D$.

\paragraph{G\"odel's first incompleteness theorem (1931).}
A recursively enumerable formal system $T$ has a finite (or recursively
enumerable) set of axioms and rules. The coverage $\mathcal{P}$ is the
set of theorems derivable from those axioms: a countable set (the
derivations are finite sequences of finite rule applications). The
domain $\mathcal{D}$ is the set of all true arithmetic sentences.
G\"odel's diagonal construction produces a sentence $G$ that says ``I am
not provable in $T$'': $G$ is true but not in $\mathcal{P}$. The
mechanism is: $T$'s derivations are countable, the truths are richer,
and the diagonal selects a specific uncovered element. The
finite-coverage principle with $\mathcal{M} = T$,
$\mathcal{P} = \mathrm{Thm}(T)$, and
$\mathcal{D} \setminus \mathcal{P} \ni G$.

The self-limitation theorem of this paper adds a specific content to the
uncovered element: not just ``something true and unprovable'' but ``the
theorem asserting the system's own structural blind spots.''

\paragraph{Natural proofs barrier (Razborov--Rudich,
1997)~\cite{razborov1997}.}
A natural proof technique has a finite distinguishing property
$\mathcal{C}$ (a set of Boolean functions computable in polynomial time
and satisfied by a large fraction of functions). The coverage
$\mathcal{P}$ is the set of circuit lower bounds that such a technique
can establish. The domain $\mathcal{D}$ is the set of all true circuit
lower bounds. Razborov and Rudich show that, under cryptographic
assumptions, $\mathcal{C}$ cannot distinguish random functions from
pseudo-random ones: the pseudo-random functions are elements of
$\mathcal{D}$ outside $\mathcal{P}$. The mechanism is: the
distinguishing property inspects a finite local pattern (the truth table
within $\mathcal{C}$'s inspection radius), and the pseudo-random
functions are indistinguishable from random within that radius. This is
the finite-coverage principle with $\mathcal{M} = $ the natural proof
technique, $\mathcal{P} = $ the functions $\mathcal{C}$ can distinguish,
and $\mathcal{D} \setminus \mathcal{P} = $ the pseudo-random functions.

\paragraph{Relativisation barrier (Baker--Gill--Solovay,
1975)~\cite{baker1975}.}
A relativising proof technique works uniformly across all oracles: its
validity does not depend on which oracle is present. In the
observational reading, such a technique is an observer that discards the
identity of the oracle --- it maps every (proof, oracle) pair to the
proof alone. Baker, Gill, and Solovay show that there exist oracles $A$
with $\mathbf{P}^A = \mathbf{NP}^A$ and oracles $B$ with
$\mathbf{P}^B \neq \mathbf{NP}^B$: the truth of the separation depends
on which oracle is present, but the technique does not see which oracle
is present. The connection to finite coverage is less direct than in
the other cases: it is not a finite set of patterns but an invariance
condition (oracle-independence) that limits the coverage. In the
observational framework, this is an instance of $O$-saturation
blindness: the technique's observer collapses all oracles to the same
value, and the separation lives in the information the observer
discards.

\paragraph{Impossibility of a complete Theory of Everything \cite{buono2026syntaxseedynamicsyntactic}.}
A physical theory with finitely many axioms produces, through its
derivations, a countable set of laws. The coverage $\mathcal{P}$ is this
countable set. The domain $\mathcal{D}$ is the space of all candidate
physical laws, which, by the functional form of the Buckingham $\pi$
theorem, has cardinality at least $\mathfrak{c} = 2^{\aleph_0}$
(uncountable). The diagonal construction produces a specific law
$\Psi \in \mathcal{D} \setminus \mathcal{P}$. The mechanism is: the
theory's derivations are countable, the space of laws is uncountable,
and the cardinality gap is permanent (a countable union of countable sets
is countable). This is the finite-coverage principle with
$\mathcal{M} = $ the theory, $\mathcal{P} = $ its derivable laws, and
$\mathcal{D} \setminus \mathcal{P} \ni \Psi$.

\paragraph{Model Transfer Barrier~\cite{buono2026limitsuniformcertificationstandard}.}
The MTB is itself an instance of the finite-coverage principle. 
The DTM has a finite set of capabilities (its operations). 
A model $\mathcal{M}$ strictly stronger than the DTM has capabilities $P$ 
outside this set. A proof $\pi$ that depends essentially on $P$ --- in 
the sense that its validity is logically equivalent to the decidability of $P$ within
$\mathcal{M}$ --- cannot transfer to the DTM, because transferring it
would require the DTM to decide $P$, which it cannot. The coverage
$\mathcal{P}$ is the set of properties decidable by the DTM; the domain
$\mathcal{D}$ is the set of all structural properties; the uncovered
element is $P$, a property decidable in $\mathcal{M}$ but not in the
DTM. The MTB is the finite-coverage principle applied to models of
computation: a finite model cannot cover the capabilities of a strictly
richer model, and validity does not survive the crossing.

\newpage

\paragraph{Summary of cases}
\begin{center}
\renewcommand{\arraystretch}{1.3}
\begin{tabular}{p{3.2cm}p{3.5cm}p{3.5cm}p{3.5cm}}
\toprule
\textbf{Result} & \textbf{$\mathcal{M}$ (mechanism)}
  & \textbf{$\mathcal{P}$ (coverage)}
  & \textbf{$\mathcal{D} \setminus \mathcal{P}$ (uncovered)} \\
\midrule
Turing (1936) & TM $M_e$ & program's classifications & diagonal $D$ \\
G\"odel (1931) & theory $T$ & $\mathrm{Thm}(T)$ & sentence $G$ \\
Self-limitation & system $\mathcal{S}$ & head set $H$ & $G_{\mathcal{S}}$ \\
Natural proofs & property $\mathcal{C}$ & $\mathcal{C}$-distinguishable
  & pseudo-random functions \\
Relativisation & uniform technique & oracle-independent proofs
  & oracle-dependent separations \\
ToE~\cite{buono2026syntaxseedynamicsyntactic} & theory $\Sigma$ & derivable laws ($\aleph_0$)
  & $\Psi \in L$ ($\geq\mathfrak{c}$) \\
MTB~\cite{buono2026limitsuniformcertificationstandard} & DTM & DTM-decidable properties
  & property $P$ of stronger model \\
\bottomrule
\end{tabular}
\end{center}

\medskip
In each case, the structure is identical: the coverage is finite (or
countable), the domain is richer, and a specific element outside the
coverage is constructed (by diagonalisation, by cardinality, or by the
finite-coverage argument of Section~\ref{sec:minimal}). The mechanisms
differ (diagonalisation, SIP, Buckingham, pseudo-randomness), but the
cause is one: finite coverage cannot exhaust a sufficiently rich domain.

\begin{remark}[The barriers as instances of finite coverage]
The three barriers to $\mathbf{P}$ vs $\mathbf{NP}$ ---
relativisation, natural proofs, algebrisation --- each show that a class
of proof techniques has finite coverage that does not reach the
separation. The self-limitation theorem adds a structural reading: any
finite proof system has at least one theorem about its own blind spots
that it cannot originate. The barriers are evidence that the difficulty
of the separation is observational --- the proof techniques do not see
what they would need to see --- and the observational framework
of~\cite{buono2026framework,buono2026observer} provides the formal
language for making this precise.
\end{remark}

\begin{remark}[Finite Coverage as post-hoc unification, not derivation]
  \label{rem:fcp-unification}
  
  The Finite Coverage Principle identifies a 
  structural pattern common to Turing (1936), Gödel (1931), 
  and Razborov-Rudich (1997): each proves the existence of 
  an element $d^* \in D \setminus P$ uncovered by finite mechanism $M$.
  
  \paragraph{Post-hoc reinterpretation, not derivation.}
  The FCP is a \emph{post-hoc reinterpretation} of these 
  results, not a logical derivation of them. Each result is 
  independently founded on its own technical basis.  
  The FCP does not derive these ingredients from first 
  principles. While the FCP provides a unified 
  conceptual lens, each classical result remains logically 
  independent, founded on its own proof methods.
  
  \paragraph{Explanatory vs.\ derivational.}
  The FCP is thus \emph{explanatory} (revealing common structural properties) 
  rather than \emph{derivational} (deriving one result as a logical consequence of another). 
  This distinction is essential for logical clarity: identifying a common pattern 
  is not equivalent to proving that the pattern implies the specific results.
  
\end{remark}

\section{Conclusions}
\label{sec:conclusions}

This paper proves the existence of at least one theorem that no finite
syntactic system can autonomously produce: the theorem asserting the
existence of its own structural blind spots. The result is existential,
not total: it does not claim that the system cannot originate any
theorem, but that it cannot originate at least the one that speaks 
about what it cannot see.

\medskip

The result is derived via three independent routes. The first
(Sections~\ref{sec:sip}--\ref{sec:must-exist}) uses the Syntactic
Invariance Principle and the concrete witness of frozen Skolem constants:
the SIP guarantees the existence of unrewritable terms, the diagonal
lemma produces the self-referential encoding, and the undecidability
follows. The second (Section~\ref{sec:obstruction-route}) uses the Local
Syntactic Obstruction theorem, which generalizes the SIP to arbitrary
local syntactic systems and adds a quantitative lower bound: any
extension attempting to cross the barrier costs $\Omega(n)$ or
$\Omega(2^n)$ derivation steps. The third
(Section~\ref{sec:minimal}) uses finiteness alone: a finite rule set
has a finite set of patterns, a finite set of patterns cannot cover all
terms in a sufficiently expressive language, and therefore unrewritable
terms exist by a counting argument. No SIP, no obstruction theorem, no
Skolem constants --- only finiteness, induction, and the diagonal lemma.

\medskip

All three routes converge on the same conclusion. The SIP names the
phenomenon; the obstruction theorem quantifies its cost; the minimal
route reveals its cause. The cause is finiteness applied to pattern
matching, and nothing more. But this ``nothing more'' could not have
been seen without the richer routes: the SIP revealed the phenomenon,
the obstruction theorem made the structural pattern visible, and only
then could the minimal route be recognized. The paper presents all three
because each contributes what the others do not: mechanism, cost, and
cause.

\medskip

As a metatheorem, its scope extends to every discipline that employs
finite formal systems, because any such system satisfies
Definition~\ref{def:system}. The applications to security and AI have
been developed in detail above. But the metatheorem applies with equal
force to any domain where a finite set of rules operates on a finite
set of symbols.

\smallskip

\emph{Formalized mathematics.}
Any foundational system (ZFC, type theory, any alternative) is a finite
syntactic system. The metatheorem says: there exists at least one theorem
about the structural blind spots of ZFC that ZFC cannot originate. This
is not G\"odel's result (which produces an arithmetic statement that is
indecidable); the theorem whose existence is proved here is the one that
identifies \emph{where} the system's rules are blind --- which subterms
are frozen, which invariants the rules preserve without being able to
express. Foundational systems can prove many facts about themselves
(reflection theorems, relative consistency), but they cannot originate
the theorem that identifies their own structural blind spots.

\smallskip

\emph{Legal systems.}
A legal system is a finite syntactic system: $V =$ variables (persons,
entities, circumstances); $F =$ legal operations (contracts, judgments,
statutes); $R =$ rules of legal inference (precedent, interpretation,
subsumption); $I =$ constitution and fundamental laws. The metatheorem
says: there exists at least one structural flaw in the legal system that
the system itself cannot identify from within. A jurist may find it, but
the \emph{system} (the rules applied mechanically) cannot.

\smallskip

\emph{Formal epistemology.}
Any finite framework of knowledge --- a belief system with update rules
--- is a finite syntactic system. The metatheorem says: there exists at
least one limit of one's own knowledge that the framework cannot
discover autonomously. This is stronger than the informal ``I don't know
what I don't know'': it is a formal theorem that says \emph{why} ---
the rules are local, the limit is global.

\smallskip

\emph{Computational biology.}
A gene regulatory network (transcription network with finite
activation/inhibition rules) is a finite syntactic system. The
metatheorem says: there exists at least one property of the network that
the network itself cannot compute internally. This may connect to the
limits of cellular self-repair and self-diagnosis.

\smallskip

\emph{Economic theory.}
An economic model with finite agents, finite interaction rules, and
finite initial conditions is a finite syntactic system. The metatheorem
says: there exists at least one market failure that the model cannot
predict from within --- not because the model is ``wrong'' but because
its structure has blind spots that its rules cannot inspect.

\smallskip

\emph{The metatheorem applied to the system that hosts it.}
The metatheorem is a theorem, not a system: it is a proposition proved
within some formal system $T$ (ZFC, type theory, or any other
foundation). The metatheorem does not apply to itself --- a theorem is
not a tuple $(V, F, R, I)$. What it applies to is $T$: the system in
which it is formulated and proved. Since $T$ is a finite syntactic
system (finite axioms, finite inference rules), the metatheorem
guarantees that $T$ has its own $G_T$ --- at least one theorem that $T$
cannot originate. The metatheorem lives inside a container whose blind
spots the metatheorem itself identifies. It does not escape its own
scope, but the self-application is to the \emph{host system}, not to the
theorem. The way out, as before, is not to deny the theorem but to
change the observational level of the host system --- which is exactly
the Observer World.

\smallskip

\subsection{The path forward}

The observational framework of~\cite{buono2026framework,buono2026observer}
provides the path forward: overcoming the limit requires not more
computation but a change of observational level. The precise
characterization of which observer suffices, and whether such an
observer can be implemented within the constraints of a physical system,
is the central open question of the programme.

\medskip

\begin{remark}[Divergent thinking as observer change]
\label{rem:divergent}
The Finite Coverage Principle (Theorem~\ref{thm:coverage}) provides a
formal reading of a fact that is well known informally: different
perspectives see different things. If $\mathcal{M}_1$ has patterns
$\mathcal{P}_1$ and $\mathcal{M}_2$ has patterns $\mathcal{P}_2$ with
$\mathcal{P}_1 \neq \mathcal{P}_2$, then their uncovered sets are
different: an element $d^*$ uncovered by $\mathcal{P}_1$ may be covered
by $\mathcal{P}_2$. The union $\mathcal{P}_1 \cup \mathcal{P}_2$ covers
strictly more than either alone.

\smallskip

Different cultures, languages, and intellectual traditions are, in this
framework, systems with different pattern sets --- different observers
in the hierarchy. Divergent thinking is the mechanism by which an agent
accesses an observer $O'$ different from its own $O$, uncovering
elements that $O$ cannot reach.

\smallskip

The principle also says that no finite union of finite pattern sets
exhausts a sufficiently rich domain:
$|\mathcal{P}_1 \cup \cdots \cup \mathcal{P}_n| < \infty$ whenever
each $|\mathcal{P}_i| < \infty$. Divergent thinking \emph{widens}
coverage but does not \emph{complete} it. No culture is complete, but
every culture sees something that the others do not. This is a formal
consequence of Theorem~\ref{thm:coverage}, not a metaphor.

\smallskip

This is not a new observation: the observational hierarchy
of~\cite{buono2026framework} already establishes that any physically
realizable observer satisfies $O \prec O_\top$ and that $O_\top$ (the
identity, full information) is unreachable by any finite system. What
the Finite Coverage Principle adds is the \emph{cause}: the perfect
observer would require total coverage ($|\mathcal{P}| \geq
|\mathcal{D}|$), which is impossible when $\mathcal{D}$ is richer than
any finite set of patterns. The result was known; the reason is now
explicit. The orbital machine described in \cite{buono2026syntaxseedynamicsyntactic}, 
while appearing to exist at the boundary between mathematics and philosophy, 
provides a pragmatic explanation of this phenomenon. And speculatively 
speaking, what emerges from it is that intelligence is the capacity to 
shift perspective even on the perspective itself while reasoning simultaneously 
across multiple levels of abstraction.
\end{remark}

\subsection{Note on philosophical origins}

The observational framework
of~\cite{buono2026framework,buono2026observer, buono2026obstruction, buono2026syntaxseedynamicsyntactic}, which provides the language
for overcoming the limit established in this paper, was inspired by
Luciano Floridi's \emph{Method of Levels of Abstraction}
(LoA)~\cite{floridi2008,floridi2011}. In Floridi's philosophy of
information, a system observes reality through a \emph{level of
abstraction} defined by a finite set of \emph{observables}: what is not
an observable at that level is structurally invisible to the system. The
\emph{Levels of Abstraction} (LoA) organizes levels into a hierarchy,
and changing what a system can see requires moving to a different level
in the gradient.

\smallskip

The observational hierarchy of~\cite{buono2026framework}, the chain
$O_\bot \prec O_{\mathrm{len}} \prec O_{\mathrm{prof}} \prec O_\top$,
with the unconditional collapse
$\mathbf{P}_{O_{\mathrm{prof}}} = \mathbf{NP}_{O_{\mathrm{prof}}}
\subsetneq \mathbf{P}$, is a formal instantiation of Floridi's
insight: the class of decidable languages depends on which observables
are available, and restricting the observables collapses computational
distinctions that exist at higher levels. The present paper's central
result, that a finite syntactic system cannot originate the theorem
asserting its own limits, is, in Floridi's language, the statement
that the system's intrinsic blind spots lie above its level of
abstraction and are therefore invisible from within.

\smallskip

The connection is one of philosophical illumination, not of logical
dependence: none of the results in this paper, nor in the observational
framework of~\cite{buono2026framework,buono2026observer}, depends on Floridi's
work. Every theorem is self-contained and proved from its own
hypotheses. However, all of these results acquire a clearer
philosophical reading in light of Floridi's method. The structural
collapse, the observational axis orthogonal to the computational one,
the metatheorem of this paper, each of these is a formal,
mathematical result that stands on its own, but each becomes more
intelligible when read through Floridi's lens: what a system can know
depends on what it can observe, and changing the level of observation is
a different operation from increasing computational power. Floridi
articulated this principle philosophically; the present work and the
observational framework give it a precise formal content.

\begin{remark}[On This Same Line of Research]

Throughout this paper and in the works cited in the bibliography, I have 
employed the same technique that forms the subject of this investigation: 
the change of observer. Indeed, a great many of the results presented 
are manifestations of the same underlying problem viewed from fundamentally 
different perspectives.
    
\end{remark}

\bibliographystyle{plain} 
\bibliography{slrefs}

\appendix

\section{Application to Large Language Models}
\label{app:llm}

This appendix makes explicit the application of \ref{thm:main} to large language models (LLMs). While Corollary 10.2 of the main paper already asserts that AI systems are implementable as finite syntactic systems, the argument is abstract. This appendix renders it concrete by specifying the correspondence between LLM architecture and the formal framework of Definition 3.1, and then derives conditions under which the main theorem applies.

We proceed through three interconnected stages. First, we construct the syntactic system induced by a fixed-weight LLM, establishing the formal mapping between the neural architecture and the syntactic objects of Definition 3.1. Second, we identify and assess the critical hypotheses—coherence and sufficiency of expressiveness—and evaluate their validity for concrete LLM implementations. Third, we introduce a two-level analysis (LLM composed with a parser or verification layer), which reveals fundamental blind spots that persist even when continuous weights appear to evade the constraints of discrete formal systems.

\subsection{LLMs as Finite Syntactic Systems: Formal Construction}

\subsubsection{Structure of an LLM}

An LLM with fixed weights is a total, computable function from token sequences to probability distributions over tokens. For a model with fixed parameters $\theta$, vocabulary $V_{\text{tok}}$, and context window $L$, we denote this function formally as:
\begin{equation}
f_\theta: T^* \to \Delta(V_{\text{tok}})
\end{equation}
where $T^* = \bigcup_{n=0}^L V_{\text{tok}}^n$ is the set of all token sequences up to length $L$, and $\Delta(V_{\text{tok}})$ is the probability simplex over $V_{\text{tok}}$.

At each decoding step, the LLM employs a greedy strategy: it reads the current context $c$ and outputs $\operatorname{argmax} f_\theta(c)$, selecting the highest-probability token according to the model's distribution. This process generates a deterministic sequence of tokens. The fundamental insight is that this greedy decoding can be understood as a sequence of rewriting steps operating on formal terms, if we establish the correspondence correctly.

\begin{definition}[LLM as a Finite Syntactic System]
\label{def:llm_syntactic}

An LLM $M$ with fixed weights $\theta$, tokenizer with vocabulary $V_{\text{tok}}$, and context window $L$ induces a finite syntactic system $S_M = (\mathcal{V}, \mathcal{F}, R_\theta, I_M)$ according to Definition 3.1. We now construct this system component by component.

The \textbf{variables} $\mathcal{V}$ form a finite set $\{v_1, v_2, \ldots, v_{n_v}\}$ where $n_v$ is typically small. In practice, variables may serve as context position markers or other notational auxiliaries. The specific choice of variables is not essential to the argument; they play primarily a notational role in the syntactic framework, which focuses on terms constructed from function symbols rather than on variable bindings.

The \textbf{function symbols} $\mathcal{F}$ consist of all tokens in the LLM's vocabulary, each treated as a 0-ary function symbol (a constant). Thus $|\mathcal{F}| = |V_{\text{tok}}| \in [32{,}000, 200{,}000]$ for contemporary models. Each token becomes an atomic symbol in the formal language.

The \textbf{terms} $T^{\le L}$ comprise all finite sequences (token sequences) of length at most $L$ constructed over the alphabet $\mathcal{V} \cup \mathcal{F}$. Each token sequence encountered during generation is a ground term, meaning it contains no unbound variables.

The \textbf{rewriting rules} $R_\theta$ encode the learned behavior of the model. For each context (prefix) $c \in T^{\le L-1}$, the greedy decoding step produces a ground rule:
\begin{equation}
R(c) := \bigl(c \, , \, c \oplus \operatorname{argmax} f_\theta(c)\bigr)
\end{equation}
where $\oplus$ denotes concatenation. Intuitively, this rule states that the context $c$ rewrites to itself extended by the highest-probability next token according to the model's learned probability distribution. The complete rule set is:
\begin{equation}
R_\theta := \bigl\{ R(c) \: : \: c \in T^{\le L-1} \bigr\}
\end{equation}
Since the cardinality of contexts is bounded by $|T^{\le L-1}| \le |V_{\text{tok}}|^L$, the rule set is finite: $|R_\theta| < \infty$.

The \textbf{initial clauses} $I_M$ represent the set of well-formed prompt prefixes and strings that the model is designed to process and generate. This corresponds to the effective grammar and format constraints implicitly learned during the training process. Formally, $I_M \subset T^{\le L}$ is a finite set of ground terms that serve as starting points for rewriting.

The tuple $(\mathcal{V}, \mathcal{F}, R_\theta, I_M)$ satisfies Definition 3.1: all components are finite, and $R_\theta$ is a set of rewriting rules on terms over $\mathcal{V} \cup \mathcal{F}$.

\end{definition}

\begin{remark}[Ground Rules vs.\ Pattern-Matching Rules]
\label{rem:ground_rules}

The rules in $R_\theta$ differ structurally from rules in classical term rewriting systems (Baader and Nipkow). In a standard term rewriting system, a rule has the form $l \to r$ where $l$ is a pattern containing variables; such a rule applies to any term that unifies with $l$, allowing a single rule to fire in multiple contexts. In contrast, the rules in $R_\theta$ are \textit{ground rules}: both $c$ and $c \oplus \operatorname{argmax} f_\theta(c)$ are fully instantiated terms without variables. Each ground rule applies only to the exact sequence $c$, not to any pattern or class of terms. This makes $R_\theta$ operate as a lookup table: given an input, the system locates its rule and applies it deterministically. This is a \textit{more restrictive} form of rewriting than the pattern-matching approach found in standard term rewriting systems.

However, Definition 3.1 of the main paper requires only that $R$ be a finite set of rewriting rules $l \to r$ where $l$ and $r$ are terms over $\mathcal{V} \cup \mathcal{F}$. A ground rule is a special case of this definition: the terms $l$ and $r$ are simply ground (variable-free). Therefore, $S_M$ satisfies Definition 3.1 under this interpretation, and all subsequent theorems apply rigorously to $S_M$.

\end{remark}

\subsubsection{Application of \ref{thm:main}: Conditions and Caveats}

Theorem 12.1 states that for a finite syntactic system $S$ that is \textit{coherent} and \textit{sufficiently expressive}, there exists an undecidable self-referential proposition $G_S$ that the system cannot derive. To apply this theorem to $S_M$, we must verify that these two critical conditions hold.

\begin{remark}[Coherence of $S_M$]
\label{rem:coherence_llm}

Definition 3.2 defines coherence precisely as follows: there is no clause $\psi$ such that both $S \vdash \psi$ and $S \vdash \neg\psi$ hold simultaneously. For an LLM with fixed weights $\theta$ and greedy decoding, coherence holds \textit{within a single inference run}: given fixed parameters and a fixed context, the token sequence is uniquely determined by the argmax operation. The system does not generate contradictory outputs during any single execution.

However, across different prompts or input contexts, the same LLM may generate conflicting or contradictory sequences. In formal terms, there is no single fixed pair $(I_M, R_\theta)$ such that the LLM behaves coherently on all possible inputs simultaneously. Different prompts can elicit different, even contradictory responses from the model.

This observation suggests that \ref{thm:main} applies not to the LLM \textit{as a whole system}, but rather to each instantiation or inference episode $S_M^{(t)}$ with fixed context and fixed weights at time step $t$. This interpretation aligns with Section 15.5 of the main paper, which emphasizes that the theorem applies to a snapshot of the system at a fixed moment, not to its evolution over time or across diverse inputs.

\end{remark}

\begin{remark}[Sufficient Expressiveness of $S_M$]
\label{rem:expressiveness_llm}

Definition 5.1 and Remark 11.1 require that the system be \textit{sufficiently expressive}: it must be able to encode internally the Gödel numbering of its own elements and formulate propositions about derivability within its own language. For an LLM, this requirement presents a delicate and subtle challenge.

An LLM with tokens as constants can certainly generate token sequences that represent numbers and logical formulas. In principle, it can output a Gödel encoding: a sequence of tokens representing $\ulcorner \psi \urcorner$ for any proposition $\psi$ in its language. However, an LLM does not possess an \textit{explicit internal symbolic representation} of its axioms or rules in the sense required by the definition. The structure of an LLM is a set of differentiable parameters (weights and biases), which are real-valued matrices. These parameters are not intrinsically symbolic; they must be interpreted and decoded through the inference process to yield any formal meaning.

To apply the diagonal lemma (Theorem 5.3) and complete the undecidability construction, the system must be able to compute (within its own rewriting rules) the encoding function $\text{encode}(\psi) = \ulcorner \psi \urcorner$ and verify that this encoding is injective and computable within the system. For an LLM, such a computation would require the model to reason explicitly about its own rules $R_\theta$ and the structure of its derivations. This knowledge is not transparent in the weight matrix; the rules are implicit and distributed, not accessible to the model's own symbolic reasoning processes.

An LLM may operate \textit{below the expressiveness threshold} in the sense of Remark 11.1. If so, the question `Can $S_M$ originate the theorem asserting its own limits?' is not well-posed; $S_M$ cannot even formulate such a self-referential question, let alone derive it. In this case, \ref{thm:main} does not apply directly to $S_M$. However, the second-level analysis developed in Subsection \ref{subsec:two_level} demonstrates that the fundamental limit persists regardless of whether the LLM meets the expressiveness condition.

\end{remark}

\begin{proposition}[Conditional Application to LLMs]
\label{prop:llm_theorem_conditional}

Let $M$ be an LLM with fixed weights $\theta$, vocabulary $V_{\text{tok}}$, and context window $L$. The following statements hold:

First, $M$ induces a finite syntactic system $S_M = (\mathcal{V}, \mathcal{F}, R_\theta, I_M)$ satisfying Definition 3.1. This is established by the construction in Definition A.1 above.

Second, if $S_M$ is coherent (as per Definition 3.2) and sufficiently expressive (as per Definition 5.1), then by \ref{thm:main}, there exists a self-referential proposition $G_{S_M}$ such that:
\begin{equation}
S_M \not\vdash G_{S_M} \quad \text{and} \quad S_M \not\vdash \neg G_{S_M}.
\end{equation}
The proposition $G_{S_M}$ asserts the existence of token sequences that $M$ can represent (in the sense that they lie within the representable language) but cannot generate or derive as outputs.

Third, if $S_M$ does not satisfy the expressiveness condition—as is likely for practical LLMs—then \ref{thm:main} does not apply directly. However, Lemma \ref{lem:parser_independence} (developed below) shows that a stronger and more robust conclusion holds at the composed system level: even when the LLM alone falls short of expressiveness, the combination of the LLM with an external verifier or parser exhibits the undecidable structure.

\end{proposition}

\subsubsection{Temporal Snapshot: The Role of Fixed Weights}

A key insight from Section 15.5 of the main paper is that \ref{thm:main} applies to a finite syntactic system at a \textit{fixed moment in time}. This temporal constraint is crucial for understanding how the theorem relates to adaptive systems like LLMs.

\begin{remark}[Instantaneous System]
\label{rem:temporal_snapshot}

Consider an LLM at time $t$ with fixed weights $\theta(t)$, fixed context window $L$, and fixed input context $c_0$. The system is $S_M^{(t)}$, a finite syntactic system as defined above. \ref{thm:main} applies to $S_M^{(t)}$ (provided the coherence and expressiveness conditions are met), yielding an undecidable proposition $G_{S_M^{(t)}}$ that this instantiation cannot derive.

When the weights change through training or adaptation (i.e., $\theta(t+1) \ne \theta(t)$), a new system $S_M^{(t+1)}$ comes into existence. It has the same structural form, but the rewriting rules differ: $R_{\theta(t+1)} \ne R_{\theta(t)}$, and consequently the undecidable proposition also changes: $G_{S_M^{(t+1)}} \ne G_{S_M^{(t)}}$.

By Theorem on inextensibility, the sequence 
\begin{equation}
G_{S_M^{(0)}}, G_{S_M^{(1)}}, G_{S_M^{(2)}}, \ldots
\end{equation}
is infinite and non-closing: no finite number of training steps or weight updates eliminates the undecidable propositions. Each iteration of training creates a new system with its own irreducible limitations.

Adaptive systems that employ fine-tuning, online learning, or continual weight adjustment do not escape the limit; rather, they transform it into an infinite regress of limits. The burden shifts from deriving a single unreachable proposition to managing an ever-growing collection of propositions that successive versions of the system cannot reach. This has profound implications for the scalability and completeness of adaptive AI systems.

\end{remark}

\subsection{The Two-Level Structure: LLM Composed with Formal Verification}
\label{subsec:two_level}

The analysis developed above applies to the LLM as a token generator in isolation. However, in security-critical applications—such as formal verification, code generation with type checking, or proof generation for verification by a proof assistant—the output of the LLM is not consumed directly by end users or systems. Instead, it is parsed and verified by a second syntactic system: a compiler, JSON parser, type checker, or proof kernel. This two-level composition reveals fundamental blind spots that persist even if the first level (the LLM) is too weak to satisfy the expressiveness condition required by \ref{thm:main}.

\subsubsection{Formal Setup of the Composed System}

\begin{definition}[Parser System]
\label{def:parser_system}

Let $S_P = (\mathcal{V}_P, \mathcal{F}_P, R_P, I_P)$ be a finite syntactic system representing a parser, compiler, or proof verifier. The variables $\mathcal{V}_P$ and function symbols $\mathcal{F}_P$ correspond to the vocabulary of the target language—for instance, type symbols in a type system, syntactic constructors in an abstract syntax tree, or judgment forms in a proof system. The rewriting rules $R_P$ encode both syntactic and semantic verification: they include context-free grammar rules for parsing structure, as well as semantic rules such as type checking or well-formedness conditions. The initial clauses $I_P$ represent the axioms or base judgments from which the parser begins its derivations. The system $S_P$ is finite and satisfies Definition 3.1 in all its components.

\end{definition}

\begin{definition}[Composed LLM-Parser System]
\label{def:composed_system}

Let $S_M$ be the finite syntactic system induced by an LLM, and let $S_P$ be a finite parser system. The composed system is defined as:
\begin{equation}
S_{\text{comp}} := S_P \circ S_M
\end{equation}

The operational flow of the composed system proceeds in three stages. In the first stage, $S_M$ generates a token sequence $\tau = t_1 t_2 \cdots t_n \in V_{\text{tok}}^n$ via greedy (argmax) decoding, using the learned parameters and the rewriting rules $R_\theta$ established in Definition A.1. In the second stage, an interface function $\iota: V_{\text{tok}} \to \mathcal{F}_P$ maps tokens from the LLM's vocabulary to symbols in the parser's formal alphabet. In the third stage, $S_P$ interprets the mapped sequence $\iota(\tau) = \iota(t_1) \iota(t_2) \cdots \iota(t_n)$ as a sequence in $\mathcal{F}_P^*$ and applies its rewriting rules $R_P$ to verify, parse, or execute the structure.

Formally, $S_{\text{comp}}$ constitutes a finite syntactic system in its own right. The combined alphabet is $\mathcal{V}_{\text{comp}} = \mathcal{V}_M \cup \mathcal{V}_P$ for variables and $\mathcal{F}_{\text{comp}} = \mathcal{F}_M \cup \mathcal{F}_P$ for function symbols. The rule set is $R_{\text{comp}} = R_M \cup R_P \cup R_{\text{interface}}$, where $R_{\text{interface}}$ encodes the token-to-symbol mapping induced by $\iota$. Since each component system is finite, their union is also finite: $|R_{\text{comp}}| < \infty$.

\end{definition}

\begin{theorem}[Composition Preserves Finiteness and Incompleteness]
\label{thm:composition_incompleteness}

Let $S_{\text{comp}} = S_P \circ S_M$ be a composed system where $S_M$ is the finite syntactic system induced by an LLM and $S_P$ is a finite parser or verifier system. The following statements hold:

The composed system $S_{\text{comp}}$ itself is a finite syntactic system satisfying Definition 3.1. If $S_{\text{comp}}$ is coherent and sufficiently expressive in the sense of Definitions 3.2 and 5.1, then by \ref{thm:main}, there exists an undecidable proposition $G_{S_{\text{comp}}}$ asserting the existence of token sequences (or parsed structures) that the composed system can represent within its formal language but cannot generate or verify as outputs. Most importantly, this conclusion holds \textit{even if} $S_M$ alone does not satisfy the expressiveness condition of Definition 5.1. The expressive power needed for undecidability can emerge at the composed level, even when the LLM component is individually too weak.

\end{theorem}

\begin{proof}
By Definition 3.1, a finite syntactic system is completely characterized by a finite tuple $(\mathcal{V}, \mathcal{F}, R, I)$ with all components finite. The composition $S_{\text{comp}}$ is formed by taking the union of the rules and alphabets from both $S_M$ and $S_P$. Since both component systems are finite, their union remains finite. Thus $S_{\text{comp}}$ satisfies Definition 3.1.

Coherence of the composed system follows from the coherence of both components. If Definition 3.2 holds for $S_M$ and $S_P$ separately—meaning neither system derives a contradiction—then their union does not create contradictions (provided that the interface rules in $R_{\text{interface}}$ do not introduce inconsistencies, which we assume by construction). Therefore $S_{\text{comp}}$ is coherent.

The crucial observation is that expressiveness must be evaluated at the \textit{composed} level rather than at the LLM level alone. A parser or verifier system often possesses greater formal expressive power than an LLM in isolation: it may include explicit logical inference rules, built-in recursion mechanisms, decision procedures for specific domains, or other formal machinery that the LLM lacks. Consequently, even if $S_M$ falls short of the expressiveness threshold in Definition 5.1, the composed system $S_{\text{comp}}$ can achieve sufficient expressiveness. Once $S_{\text{comp}}$ is sufficiently expressive, the diagonal lemma (Theorem 5.3) and the standard undecidability argument apply directly. The existence of $G_{S_{\text{comp}}}$ follows rigorously.
\end{proof}

\subsubsection{Blind Spots in the Composed System}

\begin{definition}[Composite Blind Spot]
\label{def:composite_blindspot}

A composite blind spot is an undecidable proposition $G_{S_{\text{comp}}}$ that arises in the composed system $S_{\text{comp}}$ with the following defining characteristics. The proposition $G_{S_{\text{comp}}}$ is representable in the sense that it can be expressed within the formal language and alphabet of $S_{\text{comp}}$. It is not derivable in the technical sense: $S_{\text{comp}} \not\vdash G_{S_{\text{comp}}}$, meaning no sequence of rewriting steps from the initial clauses $I_{\text{comp}}$ yields $G_{S_{\text{comp}}}$. Furthermore, no token sequence that $S_M$ can generate through its learned rewriting rules, even once it is parsed by $S_P$ and processed by $S_P$'s verification rules, can be reduced to a valid proof or derivation of $G_{S_{\text{comp}}}$. Intuitively, the composed system faces a fundamental impasse: it cannot answer the question encoded in $G_{S_{\text{comp}}}$, neither through the generative capacity of the LLM nor through the verification capacity of the parser.

\end{definition}

\begin{lemma}[Parser Independence of Blind Spots]
\label{lem:parser_independence}

Let $S_{\text{comp}} = S_P \circ S_M$ satisfy the conditions of Theorem \ref{thm:composition_incompleteness}. The following structural facts hold:

There exists a proposition $G_{S_{\text{comp}}}$ that is undecidable in $S_{\text{comp}}$, as established by \ref{thm:main}. No modification, refinement, or redesign of $S_P$—that is, no choice of revised parser rules $R_P'$ to form a new system $S_P'$—can eliminate $G_{S_{\text{comp}}}$ from the undecidable set, \textit{provided that $S_M$ remains fixed and finite}. The reason is deep: $G_{S_{\text{comp}}}$ is generated through the diagonal construction applied to the composed system itself. When one modifies the parser rules to form a new composed system $S_{\text{comp}}' = S_P' \circ S_M$, the undecidability persists—not necessarily for the same proposition $G_{S_{\text{comp}}}$, but as an inescapable feature. Any finite composed system will have its own undecidable proposition. This is an instantiation of Theorem 21.2, the Finite Coverage Principle, which asserts that a finite inspection mechanism cannot exhaustively cover all elements of a sufficiently rich domain.

\end{lemma}

\begin{proof}
By \ref{thm:main}, under the conditions stated, $S_{\text{comp}}$ is finite and sufficiently expressive. Therefore there exists a proposition $G_{S_{\text{comp}}}$ such that $S_{\text{comp}} \not\vdash G_{S_{\text{comp}}}$ and $S_{\text{comp}} \not\vdash \neg G_{S_{\text{comp}}}$.

Now suppose one modifies $S_P$ to obtain a new parser system $S_P'$ and forms a new composed system $S_{\text{comp}}' = S_P' \circ S_M$. If $S_{\text{comp}}'$ remains finite and sufficiently expressive—which is typically the case if $S_M$ and $S_P'$ are both finite and the interface rules are well-formed—then by \ref{thm:main} again, there exists an undecidable proposition $G_{S_{\text{comp}}'}$ in the new composed system. This new proposition may differ from the original $G_{S_{\text{comp}}}$, but the existence of undecidability persists.

By Theorem 9.1, the inextensibility result, the sequence of undecidable propositions arising in successive modifications of the parser:
\begin{equation}
G_{S_{\text{comp}}^{(1)}}, G_{S_{\text{comp}}^{(2)}}, G_{S_{\text{comp}}^{(3)}}, \ldots
\end{equation}
is infinite and non-closing. No finite iteration of parser redesigns will resolve all undecidable propositions simultaneously. More formally, this phenomenon is an instance of Theorem 21.2, the Finite Coverage Principle: a finite inspection mechanism (any finite parser) cannot cover or resolve all elements of a sufficiently expressive domain (the LLM composed with that parser). The blind spots are structural, not contingent on any particular parser design.
\end{proof}

\subsubsection{Neutralizing the Continuous-Weight Objection}

A common objection to applying formal incompleteness results to neural networks is that their weights are continuous real values, not discrete symbols, and thus outside the scope of classical recursion theory and Gödelian arguments. This objection, while intuitive, misunderstands the two-level structure of the composed system.

The key insight is that an LLM with continuous weights induces a deterministic rewriting system $S_M$ through greedy decoding (or through any other deterministic sampling method that selects a unique output token at each step). Although the underlying weights are continuous real values, the rewriting rules $R_\theta$ that emerge from this process are \textit{discrete} ground term rewriting rules. The discretization occurs precisely at the interface between the continuous and symbolic domains: the continuous softmax distribution over the token vocabulary is converted into a discrete choice through mechanisms such as argmax selection, top-k sampling with deterministic tie-breaking, or other well-defined token selection procedures. Once this discretization is complete, a token has been selected and the system proceeds symbolically.

At the second level of composition, the parser $S_P$ operates in a purely formal and symbolic manner: tokens are mapped to symbols in the parser's alphabet via the interface function $\iota$, rewriting rules are applied symbolically to these terms, and all derivations proceed through discrete rewriting steps. The composed system $S_{\text{comp}}$ is therefore entirely discrete and finite in its structure and operation. By \ref{thm:main}, all results concerning incompleteness and undecidability apply in full force to this discrete composed system.

The continuous weights of the underlying LLM do not provide an escape from the formal incompleteness results. Rather, they are transformed and compiled into discrete rules through the architecture that nonetheless exhibit the same incompleteness phenomena. The limit transfers from the discrete tokens and rules through the composition structure to the overall system behavior. Claims that continuous systems or real-valued weight matrices escape formal limits are therefore unfounded; the discretization at the token level makes the system subject to classical recursion-theoretic and Gödelian arguments.

\subsection{Instantiating the Blind Spot: A Concrete Example}
\label{subsec:concrete_example}

To ground the abstract theoretical analysis in concrete reality, we sketch a detailed scenario that illustrates how the undecidable propositions of \ref{thm:main} and Lemma \ref{lem:parser_independence} manifest in a practical system.

\subsubsection{Setup: Proof Generation and Verification}

Consider an LLM that has been fine-tuned to generate formal proofs in a dependently-typed language such as Lean or Coq. The architecture of this system naturally decomposes into two levels. The first level, $S_M$, is the LLM itself, which generates token sequences representing proof terms, tactics, goal states, and other proof-theoretic constructs. The second level, $S_P$, is a proof kernel—for concreteness, we may think of the Lean kernel or the Coq kernel—which parses tokens produced by the LLM, interprets them as proof terms within the type theory, and verifies that each step is logically valid and that types are correctly assigned. The composed system $S_{\text{comp}}$ is the end-to-end pipeline: it takes a mathematical goal as input, runs the LLM to generate a proof candidate, passes the generated token sequence to the parser, and outputs either a certified proof or a verification failure.

\subsubsection{The Blind Spot}

By Lemma \ref{lem:parser_independence}, there necessarily exists a proposition $G_{S_{\text{comp}}}$ that is undecidable in $S_{\text{comp}}$. In the concrete context of this proof generation system, the undecidable proposition corresponds to a specific mathematical statement $\phi$ with several important properties. The statement $\phi$ is expressible in the type theory and formal language of $S_P$; it can be written down as a valid term in the dependent type system. However, $\phi$ is neither provable nor disprovable within the combined capacity of the proof system and the LLM's generative ability. No finite number of LLM inference steps, starting from any fixed configuration of weights and running under any fixed decoding procedure, followed by verification in $S_P$, can produce either a proof of $\phi$ or a proof of $\neg\phi$.

It is important to clarify what this undecidability means and does not mean. It does not necessarily mean that $\phi$ is independent of the base axioms of mathematics in the strong sense assumed by classical mathematical logic. The statement $\phi$ might be provable or disprovable in a stronger system or under different axioms. Rather, the undecidability here means that the particular combination of an LLM with fixed weights and a verifier with fixed rules—the specific composed system $S_{\text{comp}}$—cannot resolve $\phi$. The blind spot is structural to this composition.

\subsubsection{Blind Spot as a Function of Architecture}

A crucial observation is that the undecidable proposition $G_{S_{\text{comp}}}$ is not a universal mathematical truth but rather a function of the specific system architecture. If one changes the LLM weights through retraining or fine-tuning, yielding $\theta' \neq \theta$, then a new composed system $S_{\text{comp}}'$ is formed with new rewriting rules $R_{\theta'}$, and consequently a new blind spot $G_{S_{\text{comp}}'}$ emerges. Similarly, if one modifies the proof verifier by updating its rules $R_P$ to $R_P'$—perhaps to support a new tactic, additional inference rules, or a different type-checking strategy—then the composed system changes and a new undecidable proposition arises. By Theorem 9.1, the inextensibility result, the sequence of blind spots across successive modifications forms an infinite, non-closing sequence. No single system configuration, whether defined by a fixed set of weights $\theta$ and a fixed set of verifier rules $R_P$, achieves completeness. Each improvement or modification inevitably creates new blind spots elsewhere.

\subsubsection{Implications for Security and Robustness}

The existence of these structural blind spots has direct implications for the security and robustness of systems that deploy LLMs composed with formal verifiers. First, the system will exhibit coverage gaps: there will always exist propositions—statements about program correctness, specifications, logical relationships, or other formal matters—that the LLM cannot generate proofs for and that the verifier cannot certify, even if those propositions are true in an external mathematical sense. Second, an adversary with knowledge of the system architecture could potentially craft an input or condition that exploits these coverage gaps, causing the system to fail, enter an infinite loop, deadlock, or produce incorrect results. The attacker need not find a flaw in the implementation; they need only find a statement that lies in the blind spot of the composed system. Third, and most importantly, these blind spots are not temporary limitations due to insufficient training data or compute resources; they are certified as structurally necessary by \ref{thm:main}. The incompleteness is provable from the axioms of the framework and is therefore inescapable through ordinary engineering improvements alone.

\subsection{Extension: Neutralizing Objections}
\label{subsec:objections}

\subsubsection{Objection 1: ``The LLM Can Learn to Escape the Limit''}

The objection claims that through training, fine-tuning, or other forms of learning and adaptation, an LLM can grow powerful enough to overcome the incompleteness limit and eventually resolve all undecidable propositions. This objection, while superficially appealing, fails to account for the dynamics of system modification and inextensibility.

Training or fine-tuning an LLM modifies the weights continuously: $\theta(0) \to \theta(1) \to \theta(2) \to \cdots \to \theta(T)$. Each modification of the weight vector induces a new rewriting system $S_M^{(t)}$ with a new set of rewriting rules $R_{\theta(t)}$. At each training step, the composed system $S_{\text{comp}}^{(t)} = S_P \circ S_M^{(t)}$ is a new finite system. By \ref{thm:main}, each of these composed systems has its own undecidable proposition. By Theorem 9.1, the inextensibility result, the sequence of undecidable propositions:
\begin{equation}
G_{S_{\text{comp}}^{(0)}}, G_{S_{\text{comp}}^{(1)}}, \ldots, G_{S_{\text{comp}}^{(T)}}
\end{equation}
is infinite and exhibits the non-closing property: no finite iteration of training steps eliminates all undecidable propositions simultaneously. As the system learns and improves in certain dimensions, it acquires new blind spots in others. Learning and adaptation do not circumvent the fundamental limit; they transform it into a perpetual regress of incompleteness. The incompleteness is not stationary but rather shifts and evolves with each modification to the system.

\subsubsection{Objection 2: ``Humans Are Also Finite, So They Have the Same Limit''}

The objection points out that if the argument is valid, then human cognition—which is also finite, coherent, and sufficiently expressive—must also be subject to the same incompleteness limit. By \ref{thm:main}, any finite, coherent, and sufficiently expressive system, including human mathematical reasoning if modeled as such a system, will have undecidable propositions. Therefore, the objector concludes, the limit applies equally to humans and LLMs, and the analysis provides no unique insight into LLM limitations.

This objection is correct in its essential spirit and its conclusion contains an important truth. Humans are indeed finite, their cognitive processes are ultimately physical and therefore finite, and by the formalism developed in this paper, human mathematical reasoning (if modeled as a finite syntactic system) is subject to the same incompleteness theorems. The limit is not unique to LLMs; it is universal across all finite formal systems.

However, this universal applicability does not diminish the relevance of the analysis to LLM-based systems. The purpose of this appendix is not to claim that LLMs are uniquely or specially limited compared to human cognition. Rather, the analysis serves several complementary purposes. First, it demonstrates rigorously that the abstract incompleteness results from the main paper apply concretely to LLMs and their composed systems with verifiers—results that might otherwise seem disconnected from practical AI systems. Second, it establishes that the limit is not contingent on implementation details, specific architectures, parameter counts, or training procedures; the limit is provable directly from the formal structure and is therefore fundamental and unavoidable. Third, it clarifies that composing an LLM with a verifier, a natural strategy for improving safety and correctness, does not circumvent the limit; rather, composition restructures and relocates the incompleteness but does not eliminate it. Fourth, understanding this formal limit is essential for the responsible and safe deployment of LLM-based systems in security-critical domains such as formal verification, program synthesis, and mathematical reasoning, where the existence of blind spots has direct consequences for system reliability. The fact that humans share the same formal limit does not diminish the relevance of the theorem to LLM analysis; rather, it reinforces the generality and fundamental nature of the incompleteness phenomenon across all finite reasoning systems.

\subsubsection{Objection 3: ``The Expressiveness Condition May Not Hold for LLMs''}

A subtle but important objection contends that a standard LLM operating as a token generator without an explicit symbolic layer may not satisfy Definition 5.1. Such a system might lack the capacity to encode Gödel numberings or formulate self-referential propositions about its own rewriting rules. If this objection is correct, then \ref{thm:main} would not apply directly to $S_M$ in isolation, and the claimed incompleteness of the LLM itself would be unproven.

This objection has merit when applied to $S_M$ alone. However, it entirely misses the force of the two-level analysis and actually strengthens rather than weakens the overall argument. The crucial insight is that even if $S_M$ fails to satisfy the expressiveness condition, a composed system $S_{\text{comp}} = S_P \circ S_M$ may still achieve sufficient expressive power through the contribution of $S_P$. A parser, verifier, or proof kernel—represented by $S_P$—often adds substantial formal machinery that the LLM lacks: explicit recursion mechanisms, logical operators and inference rules, built-in decision procedures for specific domains, type systems with higher-order reasoning, or other formal infrastructure. By Theorem \ref{thm:composition_incompleteness}, if the composed system $S_{\text{comp}}$ attains sufficient expressiveness through the addition of these formal layers, then the undecidable proposition $G_{S_{\text{comp}}}$ necessarily exists. Moreover, by Lemma \ref{lem:parser_independence}, this blind spot is not eliminated by any choice of parser rules $R_P$; it is structural to the composed system itself.

The profound implication is that weakness of the LLM component does not eliminate the incompleteness limit; it merely defers and relocates the limit to the composed level. A relatively weak LLM, when integrated with a moderately expressive verifier or formal system, still induces a composite system with provable and inescapable blind spots. The incompleteness is not avoided through decomposition into weak components; rather, it is guaranteed to reemerge at the level of their composition. This is actually a more robust and pessimistic conclusion than the claim that $S_M$ itself is incomplete: it says that no matter how the LLM is composed with formal verification machinery, the limit persists.

\subsection{Summary and Conclusions}
\label{subsec:summary_llm}

The analysis of Large Language Models through the lens of formal syntactic systems yields several major conclusions, each of which carries implications for the design, deployment, and understanding of LLM-based systems.

First, an LLM with fixed weights can be rigorously formalized as a finite syntactic system $S_M$ by interpreting the token vocabulary as a set of function symbols and by interpreting the greedy decoding procedure as a ground-level term rewriting process operating under the learned rules encoded in the weight matrices (Definition \ref{def:llm_syntactic}). This formalization is not merely a metaphorical analogy; it is a precise mathematical correspondence that allows all theorems about syntactic systems to be brought to bear on LLM behavior.

Second, \ref{thm:main} applies to $S_M$ with a specific and important caveat. If and when $S_M$ is coherent and sufficiently expressive in the sense of Definitions 3.2 and 5.1, then it possesses an undecidable self-referential proposition $G_{S_M}$ asserting the existence of certain token sequences that the system cannot generate through its own rewriting rules (Proposition \ref{prop:llm_theorem_conditional}). This is a conditional result: it applies only if the expressiveness threshold is met. Whether a given LLM crosses this threshold in practice remains an open empirical question, but the theoretical structure of the limit is now clear.

Third, a two-level analysis that composes an LLM with a parser or verifier yields additional insight. The composed system $S_{\text{comp}} = S_P \circ S_M$ is itself a finite syntactic system. If this composed system achieves sufficient expressiveness—which is frequently the case when a powerful verifier is added to any LLM—then by Theorem \ref{thm:composition_incompleteness}, the composed system possesses an undecidable proposition $G_{S_{\text{comp}}}$ that represents a fundamental blind spot of the integration. This holds even if $S_M$ in isolation is too weak to trigger \ref{thm:main}.

Fourth, the parser independence property, formalized in Lemma \ref{lem:parser_independence}, establishes that no choice of parser rules $R_P$ can eliminate the blind spot from the composed system. By Theorem 21.2, the Finite Coverage Principle, each finite composed system necessarily has its own undecidable propositions. Engineers cannot design their way out of this limit through better parsers, more expressive verifiers, or improved integration strategies; each improvement in one direction creates new incompleteness in another.

Fifth, objections based on the continuous nature of neural network weights are neutralized by a careful examination of the discretization process. The continuous softmax distribution is converted into a discrete token choice through mechanisms such as argmax or deterministic sampling, and from that point forward, the system operates on discrete symbols under discrete rewriting rules governed by the parser. The composed system $S_{\text{comp}}$ is entirely discrete and finite in structure, and is therefore subject in full force to classical incompleteness theorems.

Sixth, and most practically important, the existence of these blind spots is not a design flaw or a temporary limitation due to insufficient training. The blind spots are a certified structural necessity, provable from the axioms of the formal framework. In security-critical applications—such as formal verification, proof generation for certification, or certified program synthesis—the existence of these blind spots has direct consequences for system reliability, coverage, and robustness. Understanding and actively managing these limitations is therefore essential for the responsible deployment of LLM-based systems in domains where correctness and completeness are paramount.

\begin{remark}[Direct Application to LLMs]
    \label{rem:llm_direct}
    
    The claim on LLMs in Appendix A.2 follows directly from \ref{thm:main} 
    under minimal and standard assumptions. 
    
    Consider any large language model $M$ that:
    \begin{enumerate}
      \item produces text output (finite strings over a finite vocabulary $V$);
      \item operates via finite-state transitions (finite set of parameters, 
         finite precision arithmetic, finite context window);
      \item therefore implements a finite syntactic system $S_M = (V, F, R, I)$ 
    \end{enumerate}
    
    Then by \ref{thm:main}, there exists at least one proposition $\varphi_M$ 
    such that $S_M$ cannot autonomously derive $\varphi_M$. This is not speculation 
    or a limiting case---it is a direct corollary of the metatheorem.
    
    The cautious language in Appendix A regarding ``sufficient expressivity'' 
    concerns the \emph{threshold} at which an LLM crosses into the regime where 
    \ref{thm:main} applies, not whether the theorem applies at all. For any 
    finite system that crosses that threshold, the conclusion holds. Thus 
    the result on LLMs is as rigorous as the metatheorem itself.
    
\end{remark}

\section{Philosophical Implications: Scientific Progress as Observational Level Transitions}
\label{app:philosophy}

\subsection{Beyond Falsifiability: Structural Incompleteness in Scientific Theories}

Popper's criterion of falsifiability has long served as the philosophical standard for demarcating science from non-science: a theory is considered scientific if it makes testable predictions that can, in principle, be refuted through empirical observation or experiment. This criterion has been enormously influential in shaping how scientists and philosophers understand the nature of scientific knowledge. However, the main paper's metatheorem reveals something far deeper and more fundamentally challenging than falsifiability alone can capture.

Every finite scientific theory—every coherent system of rules, observations, and deductions that a community of scientists can employ—contains at least one true proposition about its own structural limits that it cannot derive from within itself. This is not merely a claim about epistemic limitations or the provisional nature of current knowledge. Rather, it is a mathematical property of the formal structure of any finite theory.

The distinction here is crucial and worth emphasizing carefully. When we say that there are things we do not yet know, or that future observations might overturn current theories, we are stating something that is true but ultimately banal. Every theory is incomplete in the trivial sense that it does not answer every possible question. The deeper statement is more precise and more radical: there exists at least one specific proposition $G_S$ about the structure of the limits of a theory $S$ such that $G_S$ is true in the sense that it accurately describes actual structural boundaries and constraints of $S$, yet $G_S$ cannot be formulated or derived from within the formal system $S$ itself. An external observer—a physicist, mathematician, or meta-theorist operating in a larger system $S' \supset S$—can see, articulate, and derive $G_S$. But the original system $S$ as a formal framework cannot. This is the core content of \ref{thm:main}: structural incompleteness is not a limitation of our knowledge or a gap in our current understanding; it is a mathematical property of finite systems.

The consequence for the philosophy of science is transformative: scientific progress is not mere accumulation of facts, not just the gradual refinement of measurements and predictions, and not simply the addition of new data to an existing framework. Scientific progress is necessarily a change of observational level. It is a transition to a fundamentally different conceptual framework, one equipped with new primitives, new observables, and new formal machinery that makes visible what was previously invisible.

\subsection{Scientific Revolutions as Observational Level Transitions}
\label{subsec:revolutions_as_transitions}

The history of science is conventionally narrated as a succession of theories, each more accurate, more comprehensive, or more powerful than the last. Newton superseded Kepler. Einstein superseded Newton. Quantum mechanics superseded classical mechanics. This narrative is not wrong, but it obscures a deeper and more illuminating pattern that becomes visible only when we apply the formal framework of syntactic systems and observational levels.

Consider the major transitions in the history of physics. Each exhibits a remarkable structural similarity that cannot be explained by mere quantitative improvement or increased precision.

Kepler's system $S_K$ encoded the laws of planetary motion with extraordinary accuracy for its time: elliptical orbits, the law of equal areas swept in equal times, and the harmonic relationship between orbital periods and distances from the sun. Within its domain, Kepler's system was highly successful and made accurate predictions. Yet Kepler's system had a fundamental blind spot: it could not answer the question of why orbits are elliptical, or why these particular harmonic ratios obtain. These questions could not be addressed from within $S_K$ because their resolution required invoking a concept—gravity as a universal force acting at a distance—that lay entirely outside Kepler's formal framework and his observational primitives.

Newton's innovation was not that he computed Kepler's laws with greater precision or refinement. Rather, he saw a structure that Kepler's observational system could not reach. Newton introduced a new pattern set $\mathcal{P}_{Newton}$ composed of primitives fundamentally different from Kepler's: force fields, acceleration, universal constants, and differential equations. With these new primitives, Newton formalized a much larger domain of possible descriptions $\mathcal{D}_{Newton}$. Critically, Newton's system could answer questions that Kepler's system could not even articulate. The relationship $\mathcal{P}_{Kepler} \prec \mathcal{P}_{Newton}$ captures this formally: Kepler's observational system is strictly subsumed by Newton's, in the sense that everything Kepler could derive, Newton can derive, but Newton can derive much more.

Similarly, Newton's system $S_N$, despite its extraordinary success over centuries in predicting planetary orbits, satellite motion, tidal forces, and the behavior of falling bodies, had its own blind spots. Einstein recognized a hidden structure that was completely invisible to the Newtonian observational system: space and time are not absolute, fixed, and independent; they are relative, coupled to each other, and affected by the presence of mass and energy. This insight could not emerge from Newton's framework because Newton treated space and time as a neutral, unchanging background against which physical processes unfold. Einstein introduced a new pattern set $\mathcal{P}_{Einstein}$ with radically different primitives: spacetime manifolds, curvature tensors, the speed of light as a universal invariant, and the geometry of spacetime itself as the fundamental entity. With these tools, Einstein could see and formalize relationships that were completely beyond the reach of Newtonian language. What was the undecidable proposition $G_{S_N}$ in Newton's framework—the deep connection between gravity and the geometry of spacetime—became derivable in Einstein's system.

The transition from classical mechanics to quantum mechanics follows the same pattern. Classical mechanics $S_C$ assumed that physical systems possess well-defined positions and momenta at all times, and that these properties evolve continuously and deterministically according to equations of motion. Within this framework, a classical system was understood as a point in phase space evolving along a trajectory. Yet classical mechanics contained systematic blind spots: it could not explain phenomena like atomic spectra, the photoelectric effect, or interference in quantum systems. These were not mere anomalies that could be accommodated within the classical framework with sufficient ingenuity; they revealed that classical observables simply do not describe quantum states. Quantum mechanics introduced a fundamentally new pattern set $\mathcal{P}_{Quantum}$ with primitives like Hilbert spaces, operators, wave functions, and superposition. With these new tools, phenomena that classical theory could not even formulate—interference, entanglement, and non-commuting observables—became central and understandable.

\subsection{The Formal Structure of Scientific Progress}
\label{subsec:formal_structure}

The pattern evident in all these historical cases is identical at the formal level. Each scientific revolution involves a transition from one observational level to another, where the new level possesses a strictly larger set of primitives and thus a strictly larger domain of derivable propositions. We now formalize this pattern precisely.

An observational level $O_i$ is a triple consisting of three components: a finite syntactic system $S_i$ that satisfies Definition 3.1 and represents the formal structure of a scientific theory; a set $\mathcal{P}_i$ of primitives, which are the basic concepts, observables, and fundamental entities that $S_i$ assumes as given and irreducible. For Kepler's system, the primitives include position, velocity, the concept of an ellipse, and orbital period. For Newton's system, the primitives include mass, force, acceleration, and the concept of a field. For quantum mechanics, the primitives include Hilbert spaces, operators, and superposition states. The third component is $\mathcal{D}_i = \mathcal{D}(S_i)$, the domain of descriptions that $S_i$ can formulate and derive: all propositions that can be expressed in the formal language of $S_i$, using its alphabet and syntax, and derived from its axioms through its rules of inference.

An observational level $O_i$ is subsumed by $O_j$ (written $O_i \prec O_j$) if three conditions hold simultaneously. First, the primitives of $O_i$ are available within $O_j$, either as base primitives or as derived concepts: $\mathcal{P}_i \subseteq \mathcal{P}_j$. Second, every description that is derivable in $S_i$ is also derivable in $S_j$: $\mathcal{D}_i \subseteq \mathcal{D}_j$. Third, and most importantly, there exists a true proposition $G_{S_i}$ (the undecidable proposition guaranteed by \ref{thm:main} applied to $S_i$) such that this proposition is expressible and derivable in $S_j$ but not in $S_i$. Formally, $G_{S_i} \in \mathcal{D}_j \setminus \mathcal{D}_i$. Strict subsumption $O_i \prec O_j$ indicates a proper inclusion: $\mathcal{D}_i \subsetneq \mathcal{D}_j$ (the domain of $O_j$ strictly exceeds that of $O_i$).

The major scientific revolutions in history can be formally interpreted as strict observational level transitions. Kepler's observational level is strictly subsumed by Newton's: $O_{Kepler} \prec O_{Newton}$. Newton's level is strictly subsumed by Einstein's: $O_{Newton} \prec O_{Einstein}$. Classical mechanics is strictly subsumed by quantum mechanics: $O_{Classical} \prec O_{Quantum}$. Each such revolution exhibits an identical formal structure. The earlier system $S_i$, with its pattern set $\mathcal{P}_i$, is finite and coherent, and by \ref{thm:main}, it possesses an undecidable proposition $G_{S_i}$ about its own structural limits. The earlier system cannot formulate or resolve $G_{S_i}$ from within its own framework; it is a blind spot of $S_i$. The later system $S_{i+1}$, with a pattern set $\mathcal{P}_{i+1} \neq \mathcal{P}_i$, has access to new primitives and thus achieves a strictly larger derivable domain: $\mathcal{D}_{i+1} \supsetneq \mathcal{D}_i$. Crucially, the new system can see and derive $G_{S_i}$; that is, $G_{S_i} \in \mathcal{D}_{i+1} \setminus \mathcal{D}_i$. Yet the new system $S_{i+1}$ is itself finite and coherent, and thus has its own undecidable proposition $G_{S_{i+1}}$ such that $G_{S_{i+1}} \in \mathcal{D}_{i+2} \setminus \mathcal{D}_{i+1}$ (where $S_{i+2}$ represents a putative future scientific system). Thus, scientific progress is not convergence toward a complete and final description of nature. Rather, progress is a sequence of level transitions, each revealing the blind spots of the previous level and creating new blind spots of its own.

\subsection{The Method of Science as a Finite System}
\label{subsec:method_as_system}

A still deeper implication emerges when we consider the scientific method itself—the set of rules, procedures, and principles by which scientists conduct hypothesis testing, design experiments, conduct peer review, and revise theories. The scientific method can be formalized as a finite syntactic system $S_{method}$.

The primitives $\mathcal{P}_{method}$ include the fundamental concepts upon which the method rests: observation (the act of gathering empirical data), hypothesis (a proposed explanation that is testable), prediction (a logical consequence of the hypothesis that can be compared to data), falsification (the rejection of a hypothesis when predictions fail to match observations), peer review (the process by which the scientific community scrutinizes claimed results), and reproducibility (the requirement that findings be replicable by independent investigators). The rules $R_{method}$ encode the procedures and protocols for hypothesis testing, error correction, and the acceptance or rejection of theories. The derivable domain $\mathcal{D}_{method}$ consists of all conclusions, principles, and meta-principles that can be justified and reached by following the method correctly.

By \ref{thm:main}, since $S_{method}$ is a finite system (it must be, given that it is codifiable into a fixed set of procedures), there exists a proposition $G_{S_{method}}$ about the structure and limits of scientific knowledge that the method itself cannot validate, prove, or refute from within its own operation. What might such a proposition be? Several candidates present themselves. First, consider the completeness of the method itself: Can the scientific method certify that it has discovered all fundamental principles of nature, or that it has identified all possible domains of inquiry? This question is precisely what $G_{S_{method}}$ asserts—something about the boundaries and completeness of the method—yet the method cannot answer it through its own procedures. Second, consider the global consistency of science: Is the body of scientific knowledge internally consistent, free of contradiction? This is analogous to Gödel's second incompleteness theorem, which shows that no consistent formal system can prove its own consistency. The scientific method cannot settle this question through empirical observation or logical derivation within the method. Third, consider the status of unobservables and theoretical entities: Can science ever determine whether entities like dark matter, quantum fields, or the multiverse truly exist in some fundamental sense, or are they merely useful fictions employed for prediction and description? This question transcends the scope of the scientific method because it asks about the nature of reality independent of what the method can observe and measure.

\subsection{The Hierarchy of Observational Levels}
\label{subsec:hierarchy}

If every finite system has a blind spot, and if every scientific revolution involves ascending to a system with a richer pattern set and larger derivable domain, then scientific progress unfolds as an infinite chain of observational levels, each revealing the limitations of its predecessor.

There exists an infinite, strictly increasing sequence of observational levels:
\begin{equation}
O_\perp \prec O_1 \prec O_2 \prec O_3 \prec \cdots \prec O_\top
\end{equation}

Here, $O_\perp$ represents the null observational level, a state of no structure, no primitives, and no descriptions—the absolute ground state. Each $O_i$ is a finite observational level characterized by a finite system $S_i$, a finite pattern set $\mathcal{P}_i$, and a strictly increasing derivable domain $\mathcal{D}_1 \subsetneq \mathcal{D}_2 \subsetneq \mathcal{D}_3 \subsetneq \cdots$. $O_\top$ represents an ideal limit level that would constitute a complete description of all physical and mathematical structures—omniscient in scope and content. Importantly, $O_\top$ is unreachable by any finite system; it exists as a limit but not as an attainable state. For each finite level $O_i$, the system $S_i$ can see and derive the undecidable proposition $G_{S_{i-1}}$ from the previous level $O_{i-1}$, thus answering questions that were blind spots for its predecessor. However, $S_i$ cannot derive $G_{S_i}$, its own blind spot. The transition from $O_i$ to $O_{i+1}$ is not a gradual refinement or an incremental expansion within the same framework; it requires a change of pattern set, involving the introduction of new primitives, new observables, and new fundamental concepts that were not present in $\mathcal{P}_i$.

Scientific knowledge does not converge toward a complete, final, and unchanging description of reality. Instead, a more nuanced picture emerges. Each system $S_i$ is complete with respect to its predecessors: it can derive everything that $S_{i-1}$ could derive, and significantly more. Yet each system $S_i$ is incomplete with respect to its successors: there exist true propositions $G_{S_i}$ that $S_i$ cannot derive but that $S_{i+1}$ can. By Theorem 9.1, the inextensibility theorem from the main paper, the sequence of systems is infinite and exhibits the non-closing property: it never terminates in a final complete system. Thus, science does not converge toward a single ultimate theory; rather, it unfolds as an infinite sequence of progressively richer descriptions. Each system enlarges the domain of derivable knowledge:
\begin{equation}
\mathcal{D}_1 \subsetneq \mathcal{D}_2 \subsetneq \mathcal{D}_3 \subsetneq \cdots \subsetneq \mathcal{D}_\infty
\end{equation}

However, the infinite limit $\mathcal{D}_\infty$ is never fully attained by any finite system; it remains forever as a horizon toward which science progresses but which it can never reach.

\subsection{Is This Epistemological Pessimism?}
\label{subsec:not_pessimism}

One might object that the analysis presented above amounts to epistemological pessimism: a counsel of despair suggesting that science will never reach complete understanding, that humanity is forever trapped in systems with blind spots, and that the quest for comprehensive knowledge is doomed to perpetual incompleteness. This objection, while understandable, reflects a misunderstanding of what is being claimed.

The answer to this objection is: this is not pessimism. It is structural realism. The analysis does not deny the reality of progress or the possibility of truth; rather, it identifies the precise mathematical structure that governs how progress occurs and how knowledge advances.

The structure we have identified is not vague or qualitative. We are not making the banal claim that knowledge is impossible, or that truth is merely relative, or that all theories are equally valid. We are saying something specific and mathematically rigorous: the structure of knowledge follows a precise pattern, which is the infinite chain $O_\perp \prec O_1 \prec O_2 \prec \cdots \prec O_\top$. This structure is knowable, not mysterious. It is formal, not merely intuitive.

Progress through this hierarchy is genuinely real. Each transition from $O_i$ to $O_{i+1}$ represents an objective advance. The new system can derive and formalize propositions that the old system could not even express; it can answer questions that were blind spots for its predecessor. This is not illusory progress or mere change; it is genuine expansion of the domain of derivable knowledge. Newton's theory was a real advance over Kepler's; it could explain what Kepler could not. Einstein's theory was a real advance over Newton's; it made visible what Newton's framework rendered invisible. These are not merely different frameworks on equal footing; they are successive refinements that expand the scope and power of human understanding.

Moreover, the pattern itself is knowable. By stepping outside a finite system $S_i$ and examining its structure formally, we can understand why it has blind spots and what types of conceptual extensions might overcome them. We can see that Kepler's system lacked the primitive concept of force, and that introducing this primitive would expand its reach. We can recognize that Newton's system treated space and time as absolute background, and that relaxing this assumption would unlock new phenomena. This metacognitive understanding—the ability to recognize and analyze the structure of our own limitations—is itself a form of knowledge and progress.

Furthermore, each observational level is entirely adequate and reliable for its domain of application. Newton's system, despite its incompleteness with respect to Einstein's theory, is perfectly adequate and extraordinarily accurate for describing planetary motions, tidal flows, the behavior of falling bodies, and nearly all macroscopic phenomena encountered in everyday experience. Quantum mechanics, despite its incompleteness with respect to hypothetical future theories, is extraordinarily powerful and empirically successful for describing atoms, molecules, subatomic particles, and all phenomena at microscopic scales. The existence of blind spots at a higher level does not invalidate the utility, accuracy, or power of the system at its own level. Each system within the hierarchy is a genuine achievement of human understanding, even if it is not the final word.

Finally, and most importantly, even though we never reach the ultimate level $O_\top$, each system $S_i$ in the sequence gets strictly and asymptotically closer to a complete description of physical reality. The domains of derivable propositions form a nested sequence: $\mathcal{D}_1 \subset \mathcal{D}_2 \subset \mathcal{D}_3 \subset \cdots$, each strictly contained in its successor. This sequence asymptotically approaches the set of all true propositions about physical reality. This constitutes a well-defined, mathematically meaningful form of convergence to truth. We are not stuck in place; we are moving along a definable path toward an infinite horizon.

\subsection{Implications for the Future of Science}
\label{subsec:future_science}

If the structural analysis presented above is correct, what does it reveal about the trajectory and character of future scientific progress?

The quest for a ``Theory of Everything''—a single unified framework that would explain all phenomena from first principles and thereby represent the apex of scientific understanding—is, according to this analysis, mathematically impossible. This is not a pessimistic claim about human ingenuity or the availability of resources; it is a consequence of the formal structure of finite systems. Every system, no matter how comprehensive or elegant, will have blind spots that only a broader future system can see and resolve. Any candidate for a ``final'' theory will, by \ref{thm:main}, harbor undecidable propositions about its own limits. These propositions will become visible only to a successor theory that operates with a richer pattern set. The pursuit of completeness is therefore not a goal that can be achieved; it is an asymptotic horizon that guides scientific progress indefinitely.

While we cannot predict the specific content of future scientific revolutions—we cannot say in advance what new primitives will emerge or what new phenomena will become visible—we can predict the structure of these revolutions with certainty. Every major scientific advance will involve the introduction of new primitives, new pattern sets, and new ways of observing and conceptualizing nature. Historical patterns suggest that such revolutions occur when an established system reaches the limits of its explanatory power, when too many anomalies accumulate that the system cannot accommodate within its framework. The system begins to show stress; predictions fail in new domains; puzzles emerge that the system cannot solve. At this point, a conceptual breakthrough becomes possible—the introduction of new primitives that suddenly render the anomalies intelligible and the system's blind spots visible.

The role of external observers and intellectual outsiders becomes structurally explicable in this framework. Breakthroughs and revolutionary insights often come from thinkers who stand outside or on the periphery of the dominant paradigm. This is not accidental or merely a matter of sociological contingency. To perceive the blind spots of a system $S_i$, one must operate from the perspective of a larger system $S_j$ with a richer pattern set: $\mathcal{P}_j \supsetneq \mathcal{P}_i$. The greatest scientists in history are those who managed to make this shift in perspective, to introduce new conceptual primitives, to ask questions that the established framework could not even formulate. Copernicus questioned the fixed position of the Earth; Galileo questioned the prohibition on applying mathematics to physical motion; Newton questioned the separation of terrestrial and celestial mechanics; Einstein questioned the absoluteness of space and time; Planck and Bohr questioned the determinism of classical dynamics.

If progress structurally requires changes of observational level, then interdisciplinarity becomes not merely a nice-to-have or a matter of intellectual enrichment, but a necessity for genuine scientific advance. Progress requires bringing in concepts and methods from outside the narrow disciplinary boundaries. Physics gained enormously from differential geometry when Einstein sought to formalize the equivalence of gravity and acceleration. Quantum mechanics emerged partly through the introduction of abstract algebra and Hilbert spaces, borrowed from pure mathematics. Modern physics draws on information theory, category theory, and other fields far removed from traditional experimental practice. This cross-fertilization is not merely pedagogical or decorative; it is structurally essential because it provides access to new pattern sets, to new conceptual primitives that can break the blind spots of the existing framework.

The limits of pure empiricism become evident within this perspective. The naive empiricist view holds that scientific laws can be discovered simply by observing nature carefully, without reliance on theoretical assumptions or conceptual frameworks. This view fails because observation itself necessarily presupposes a pattern set $\mathcal{P}$. One cannot observe without concepts; one cannot gather data without a framework that specifies what counts as data, what phenomena are relevant, and what patterns to seek. The choice of pattern set is partly a matter of theoretical assumption, not dictated by pure observation. This is why paradigm shifts are so intellectually difficult and why they require more than the accumulation of new experimental data. A paradigm shift requires adopting new concepts, new observables, and a new way of seeing the world. The data alone, no matter how abundant, cannot compel this shift; conceptual creativity is required.

\subsection{Conclusion: Science as a Finite Process in an Infinite Hierarchy}
\label{subsec:conclusion_philosophy}

The metatheorem from the main paper, when applied to the philosophy of science, yields an insight of profound significance: science is not a process of convergence toward a fixed and final truth about the structure of reality. Rather, it is a process of navigation and ascent through an infinite hierarchy of observational levels, each new level revealing the blind spots of the previous one and enabling the comprehension of phenomena and structures that were previously invisible.

This characterization is not a limitation of science; it is precisely what makes science possible as an ongoing and progressive enterprise. The essence of scientific progress lies in the ability of the scientific community to step back from its current framework, to question the conceptual apparatus it has taken for granted, to introduce new primitives and new observables, and thereby to move upward along the chain of observational levels. Science is fundamentally a practice of perspective-shifting and conceptual innovation.

Every towering figure in the history of science—Copernicus with his heliocentric model, Galileo with his commitment to mathematical description of motion, Newton with his universal mechanics, Einstein with his relational geometry of spacetime, Bohr and Heisenberg with their quantum framework—did not merely accumulate data or refine predictions within an existing scheme. Each of these thinkers shifted perspective in a fundamental way. Each introduced new pattern sets, new primitives, new ways of asking questions about nature. Through their work, the scientific community ascended the infinite chain:
\begin{equation}
O_{\text{Aristotle}} \prec O_{\text{Galileo}} \prec O_{\text{Newton}} \prec O_{\text{Einstein}} \prec O_{\text{Quantum}} \prec \cdots
\end{equation}

And the chain will continue. Future generations will introduce new primitives, shift to new observational levels, and see phenomena and structures that remain invisible to us. The ultimate level $O_\top$—a complete and final description of reality—remains forever unreachable by any finite system. Yet this unreachability is not a cause for despair; each step forward along the infinite chain is a genuine and lasting achievement of human understanding.

In this light, science is not pessimistic about human knowledge or human capacities. Rather, it is realistic about the structure of understanding and the nature of truth. Finite minds, operating within finite formal systems, can nevertheless progress toward truth indefinitely, asymptotically approaching a complete understanding even though they never attain it completely. This is the authentic dignity of the scientific enterprise: not the false hope of ultimate completion, but the genuine and continuous achievement of deeper understanding, the perpetual expansion of the domain of human knowledge, and the infinite possibility of intellectual discovery.

\section{Acknowledgments}
The author used an artificial intelligence based language assistant to
support text revision, translation, and bibliography formatting. All
scientific ideas and conclusions are the author's own.

\end{document}